\documentclass[a4paper,oneside,11pt]{article}
\usepackage{graphicx}
\usepackage[utf8]{inputenc}
\usepackage[T1]{fontenc}
\usepackage{longtable,geometry}
\usepackage{authblk}
\usepackage{amsmath}
\usepackage{yhmath}
\usepackage{amssymb}
\usepackage{tikz}
\usepackage{pgfplots}
\pgfplotsset{compat=1.18}
\usepackage{dsfont}
\usepackage{braket}
\usepackage{array}
\usepackage{graphicx}
\usepackage[xindy]{glossaries}
\usepackage{fontawesome5}
\usepackage{tikz-cd}
\usepackage[babel]{csquotes}
\usepackage[doi=false,isbn=false,url=false,eprint=false]{biblatex}
\usetikzlibrary{intersections, pgfplots.fillbetween}
\usepackage{amsthm}
\usepackage{cancel}
\MakeAutoQuote{«}{»}
\usepackage{amsthm}
\theoremstyle{plain}
\newtheorem{prop}{Proposition}[section]
\newtheorem{thm}{Theorem}[section]
\newtheorem{lemma}{Lemma}[section]
\newtheorem{corollary}{Corollary}[section]

\theoremstyle{definition}
\newtheorem{definition}{Definition}[section]
\theoremstyle{remark}
\newtheorem*{rmk}{Remark}

\usepackage{pstricks,pstricks-add,pst-node,pst-tree}
\usepackage[nottoc]{tocbibind}
\usepackage[english]{babel}
\usepackage{hyperref}
\usepackage[all]{hypcap}

\NewDocumentCommand{\KD}{}{Kirkwood-Dirac }
\NewDocumentCommand{\KDpos}{}{\mathcal{E}_{\mathrm{KD}+}}
\NewDocumentCommand{\KDpospure}{}
{\mathcal{E}_{\mathrm{KD}+}^{\mathrm{pure}}}
\NewDocumentCommand{\KDpospuregen}{}
{\mathcal{E}_{\mathrm{KD}+}^{\mathrm{pure,gen}}}
\NewDocumentCommand{\Wpos}{}{\mathcal{E}_{\mathrm{W}+}}
\NewDocumentCommand{\Wpospure}{}{\mathcal{E}_{\mathrm{W}+}^{\mathrm{pure}}}
\NewDocumentCommand{\KDr}{}{V_{\mathrm{KDr}}}
\NewDocumentCommand{\Def}{m}{\emph{#1}}
\NewDocumentCommand{\proj}{m}{|#1\rangle\langle#1|}

\NewDocumentCommand{\tr}{}{\mathrm{Tr}}

\newacronym{lca}{LCA}{locally compact abelian}
\newacronym{sclca}{SCLCA}{second countable locally compact abelian}

\title{Characterizing the Kirkwood-Dirac positivity on second countable LCA groups}
\author[1]{Matéo Spriet}
\affil[1]{Univ. Lille, CNRS, Inria, UMR 8524, Laboratoire Paul Painlevé, F-59000 Lille, France}
\date{\today}

\begin{document}

\maketitle

\begin{abstract}
    We define the \KD quasiprobability representation of quantum mechanics associated with the Fourier transform over second countable locally compact abelian groups. We discuss its link with the Kohn-Nirenberg quantization of the phase space $G\times \widehat{G}$. We use it to argue that in this abstract setting the Wigner-Weyl quantization, when it exists, can still be interpreted as a symmetric ordering. Then, we identify all generalized (non-normalizable) pure states having a positive \KD distribution. They are, up to the natural action of the Weyl-Heisenberg group, Haar measures on closed subgroups. This generalizes a result known for finite abelian groups. We then show that the classical fragment of quantum mechanics associated with the \KD distribution is non-trivial if and only if the group has a compact identity component. Finally, we
    provide for connected compact abelian groups a complete geometric description of this classical fragment.
\end{abstract}

\section{Introduction}

Quasiprobability representations of quantum mechanics have been defined and studied since the early stages of the theory as an approach to understand the differences between classical and quantum systems. While introduced as purely theoretical objects, the recent progress in quantum information and quantum computing has shown that they are important tools to study the various ``quantum advantages'' that can be provided by quantum computers over their classical counterparts \autocite{KD_review}, \autocite{Pashayan_Wallman_Bartlett}.

We recall that a quantum system is modeled as follows: consider a separable Hilbert space $\mathcal{H}$. The quantum states of the system are density matrices, i.e. positive trace class operators on $\mathcal{H}$ having trace $1$, while the observables are (bounded) self-adjoint operators on $\mathcal{H}$. The model is probabilistic, and the expectation value of the observable $A$ in the state $\rho$ is given by
\begin{equation}
    \langle A\rangle_{\rho}=\tr(\rho A).
\end{equation}
A quasiprobability representation of such a quantum system is composed of a measurable space $(\Lambda,\mathcal{T}_{\Lambda})$ together with a pair of continuous linear maps $(\mu,\xi)$, such that $\mu$ associates to any trace class operator a complex measure of finite variation on $\Lambda$, while $\xi$ associates to any bounded operator a bounded measurable function on $\Lambda$.  It is asked that those maps satisfy $\xi[1]=1$ and that the overlap formula holds:
\begin{equation}
\label{overlap}
    \tr(\rho A)=\int_{\Lambda}\xi[A](\lambda)d\mu[\rho](\lambda).
\end{equation}
For a quantum state $\rho$, it follows that $\mu[\rho](\Lambda)=1$, which is why we call $\mu[\rho]$ a quasiprobability distribution. The function $\xi[A]$ is referred to as the symbol of the bounded operator $A$.
A fundamental result \autocite{Ferrie} states that a quasiprobability representation of quantum mechanics cannot satisfy both of the following conditions:
\begin{enumerate}
    \item For any quantum state $\rho$, $\mu[\rho]$ is a probability measure on $\Lambda$.
    \item For any bounded positive self-adjoint operator $A$, $\xi[A]$ is a positive function on $\Lambda$.
\end{enumerate}
This result shows that quantum mechanics is intrinsically different from classical mechanics, as its model differs from any usual probabilistic model. However, for a fixed quasiprobability representation of quantum mechanics, an interesting problem is to determine the set of quantum states $\rho$ for which $\mu[\rho]$ is a probability measure, and the set of (positive) observables $A$ for which $\xi[A]$ is a real (positive) function, which can be thought of as a ``classical fragment'' of the quantum system, with respect to the chosen quasiprobability representation. It has been shown that any algorithm running on a quantum computer that uses only states and measurements that are all in this classical fragment, as well as unitary evolutions that preserve positivity of $\mu$, can be simulated efficiently using a classical computer and the representation itself \autocite{Pashayan_Wallman_Bartlett}. In particular, if one wants to obtain a quantum advantage, one needs to ``get out'' of this classical fragment. This is the reason why, for a fixed quasiprobability representation, it is useful to describe the associated classical fragment as precisely as possible. This problem is in general not easy to tackle. For example for the Wigner distributions, while the  pure states that belong to the classical fragment are characterized in most cases \autocite{Hudson, Gross_Wigner, Crann, NicolaRiccardi2025}, characterizing the mixed states in this fragment remains an open problem, despite efforts and progress \autocite{Bröcker_Werner, Mandilara_Karpov_Cerf, Cerf_Van_Herstraeten_Beam_splitter}. The goal of this article is to address the problem of characterizing the classical fragment for the Kirkwood-Dirac representation associated with the Fourier transform over \gls{sclca} groups.

The first and most famous of quasiprobability representations is the Wigner-Weyl representation introduced in 1932 \autocite{Wigner}, which is still intensively studied today. Since the 80's, the Glauber-Sudarshan $P$ function \autocite{Cahill_Glauber,Sudarshan} and the Husimi $Q$ function \autocite{Husimi} have also played important roles in the study of quantum optical systems. Those representations are defined over the Hilbert space $L^2(\mathbb{R}^n)$ using the position and momentum operators, which are related via the Fourier transform. In 1933, J. Kirkwood defined another distribution associated to the position and momentum operators \autocite{Kirkwood}. In 1945, P.A.M. Dirac independently introduced a generalization of this distribution, associated with arbitrary pairs of non-commuting operators over any Hilbert space \autocite{Dirac}. This distribution is now called the Kirkwood-Dirac distribution. It is also known in signal analysis as the Rihaczek distribution \autocite{Rihaczek}. The purpose of P.A.M. Dirac was to construct functions of two non-commuting operators $A$ and $B$ by choosing an ordering. He considers polynomials $p(a,b)=\sum_{i,j}p_{i,j}a^ib^j$, $(a,b)\in \sigma(A)\times \sigma(B)$, where $\sigma(A)$ and $\sigma(B)$ are the spectra of $A$ and $B$, respectively. He then defines the operator $p(A,B)=\sum_{i,j}p_{i,j}A^iB^j$. For example the polynomial function $(a,b)\mapsto ab=ba$ is mapped to the operator $AB$, which in general differs from the operator $BA$. In that sense there is a choice of ordering. It can be checked that the \KD distribution (associated with $A$ and $B$) of the operator $p(A,B)$ is indeed the polynomial $p(a,b)$ \autocite{KD_review}. When $A$ and $B$ are the position and momentum operators respectively, this is known as the standard ordering, or Kohn-Nirenberg quantization \autocite{Kohn_Nirenberg_original,Kohn_Nirenberg_quantization, Folland}. We discuss in section \ref{Ordering} the relation between the \KD distribution and the Kohn-Nirenberg quantization in the setting of \gls{sclca} groups, recovering the definition given in \autocite{Kohn_Nirenberg_quantization} and in \autocite{FeichtingerKozek98}. We also argue that, when it exists, the Wigner-Weyl distribution defined in \autocite{Crann} interpolates between the \KD distribution and its complex conjugate, and that the associated quantization can therefore be interpreted as a symmetric ordering.

More recently, further finite dimensional quasiprobability representations were introduced, such as the discrete \KD distribution (see \autocite{KD_review} for a review), the Margenau-Hill distribution \autocite{Margeneau_Hill} and the Gross-Wigner function \autocite{Gross_Wigner}, in the framework of quantum information and quantum computing. The Gross-Wigner function is meant to provide a finite dimensional equivalent to the original Wigner function mentioned above. It is constructed using the discrete Fourier transforms over the finite abelian groups $\mathbb{Z}/d\mathbb{Z}$. However, its definition is valid only if $d$ is an odd number. In particular, it cannot be used to study systems of $n$ qubits, which are at the core of quantum computing and which would correspond to the abelian groups $(\mathbb{Z}/2\mathbb{Z})^n$. Recently, the \KD representation associated with the Fourier transform on finite abelian groups was introduced and studied in \autocite{KD_finite_groups}. It is meant to be a finite dimensional equivalent to the \KD representation associated with position and momentum operators. As opposed to the Gross-Wigner representation, this definition comes without any further restriction regarding the structure of the groups. This allowed the authors of \autocite{KD_simulability} to study some simulability properties of systems of $n$ qubits using the \KD distribution associated with the Fourier transform over the groups $(\mathbb{Z}/2\mathbb{Z})^n$. Furthermore, the Wigner-Weyl representation has been defined and studied in the general setting of \gls{sclca} groups \autocite{Hennings, Crann, NicolaRiccardi2025}. Again in this framework, the definition is valid only if the doubling map $g\mapsto g+g$ is invertible. 
In the present article, we extend the definition of the \KD representation associated with the Fourier transform to the abstract setting of \gls{sclca} groups, without any further restriction. This definition is very general, as it covers the original case introduced by Kirkwood, the case of finite abelian groups, as well as the $n$-dimensional torus or the $p$-adic numbers, for example.

For a \gls{sclca} group $G$, we fix a Haar measure and consider the Hilbert space $L^2(G)$. We recall in section \ref{Preliminaries} the basic group theoretical notions needed in this paper. To extend the class of objects for which we can define a \KD distribution, we introduce the notion of generalized operators. We denote by $\mathcal{S}(G)$ the set of Schwartz-Bruhat functions on $G$ and by $\mathcal{S}'(G)$ the set of tempered distribution on $G$.

\begin{definition}
\label{generalized operator}
    A \emph{generalized operator} is a continuous linear map $A:\mathcal{S}(G)\to\mathcal{S}'(G)$ represented by a distributional kernel $k_A\in \mathcal{S}'(G\times G)$ such that for any $\phi_1,\phi_2\in \mathcal{S}(G)$ we have 
    \begin{equation}
        \label{integral kernel}
        (A\phi_1)(\phi_2)=\langle k_{A},\overline{\phi_2}\otimes\phi_1\rangle=\int_{G\times G}\overline{\phi_2(g)}k_A(g,g')\phi_1(g')d\mu_G(g)d\mu_G(g').
    \end{equation}
\end{definition}
The integral expression in \eqref{integral kernel} is an abuse of notation that we will use systematically.

We say that a generalized operator $A$ is symmetric if it satisfies
\begin{equation}
    (A\phi_1)(\phi_2)=\overline{(A\phi_2)(\phi_1)}
\end{equation}
for any $\phi_1,\phi_2\in \mathcal{S}(G)$.

It is clear that any Hilbert-Schmidt operator $A$ on $L^2(G)$ is a generalized operator, as it can be represented by a square integrable kernel $k_A\in L^2(G\times G)$. To any generalized operator, we associate a \KD distribution, which is a tempered distribution on the phase space $G\times \widehat{G}$.

\begin{definition}
\label{KD def}
    Let $A$ be a generalized operator with distributional kernel $k_A$. The \emph{\KD distribution} of $A$ is the tempered distribution on $G\times \widehat{G}$ defined by
    \begin{equation}
        \mathrm{KD}_{A}(g,\chi)=\overline{\chi(g)}\int_{G}k_{A}(g,g')\chi(g')d\mu_G(g'),
    \end{equation}
    to be understood in the sense of distributions.
\end{definition}
This definition is a straightforward extension of the known definitions of the \KD distribution for $\mathbb{R}^n$ and for finite abelian groups. It appears in the literature on signal analysis as the Rihaczek distribution \autocite{Rihaczek} and was considered for elementary locally compact abelian groups in \autocite{FeichtingerKozek98}.

A generalized operator $A$ is called a \Def{generalized pure state} if it is symmetric and has rank one. It is easily checked that the kernel of a generalized pure state is of the form
\begin{equation}
    k_A=\psi\otimes \overline{\psi}
\end{equation}
for some $\psi\in \mathcal{S}'(G)$. In that case, we write in Dirac notation $A=\proj{\psi}$.
The \KD distribution of a generalized pure state $\proj{\psi}$ has the simple form
\begin{equation}
\label{KD pure}
    \mathrm{KD}_{\psi}(g,\chi)=\overline{\chi(g)}\psi(g)\overline{\hat{\psi}(\chi)},
\end{equation}
where $\hat{\psi}$ is the Fourier transform of $\psi$.

\begin{definition}
\begin{itemize}
    \item A symmetric generalized operator is called \Def{\KD positive} (or simply KD positive) if its Kirkwood-Dirac distribution is a positive measure on $G\times \widehat{G}$.
    \item A symmetric generalized operator $A$ is called \Def{\KD real} (or KD real) if its \KD distribution is a real tempered distribution, meaning that for any real valued $f\in \mathcal{S}(G\times \widehat{G})$, $\langle \mathrm{KD}_A,f\rangle\in \mathbb{R}$.
    \end{itemize}
\end{definition} 

Our main goal is to study the classical fragment of quantum mechanics associated with the \KD distribution. This article contains complete characterizations of
 \begin{itemize}
    \item $\KDpospuregen$, the set of KD-positive generalized pure states,
    \item $\KDpospure$, the set of KD-positive pure states.
\end{itemize}
We then provide, depending on the topological properties of the considered \gls{sclca} group $G$, partial or complete descriptions of
\begin{itemize}
    \item $\KDpos$, the convex set of KD-positive states,
    \item $\KDr$, the real vector space of KD-real Hilbert-Schmidt observables.
\end{itemize}

The first result, proved in section \ref{Main result}, is a description of the set $\KDpospuregen$ that holds in full generality:

\begin{thm}
\label{KD positive on second countable LCA groups}
    Let $G$ be a \gls{sclca} group. The generalized pure state $\proj{\psi}$ associated with $\psi\in \mathcal{S}'(G)\setminus\{0\}$ has a positive \KD distribution if and only if there exists a closed subgroup $H\subset G$ such that, up to the action of the Weyl-Heisenberg group, $\psi$ is a Haar measure on $H$.
\end{thm}

This theorem is a natural generalization of a result of \autocite{KD_finite_groups}. The proof has a similar underlying strategy, but is more technical as one has to make distinctions between functions, measures and distributions, notions which merge together in the case of finite abelian groups. 

The second main result, that we prove in section \ref{2nd result}, states simple topological
criteria on the group G that are equivalent to fact that the classical fragment
of quantum mechanics associated with the Kirkwood-Dirac distribution is nontrivial.

\begin{thm}
\label{existence of KD real}
    Let $G$ be a \gls{sclca} group and let $G_0$ be its identity component. The following are equivalent:
    \begin{enumerate}
        \item $G_0$ is compact.
        \item $G$ admits a compact open subgroup.
        \item $\KDpospure\neq \emptyset$.
        \item $\KDpos\neq \emptyset$.
        \item $\KDr\neq \{0\}$.
    \end{enumerate}
\end{thm}

As a consequence of this result, any group of the form $\mathbb{R}\times H$, where $H$ is some \gls{sclca} group, does not admit any non-zero KD-real Hilbert-Schmidt observable nor any KD-positive state. Hence there is no classical fragment of quantum mechanics with respect to the Kirkwood-Dirac distribution in these cases.
This applies in particular for $G=\mathbb{R}$, for which the \KD distribution is referred to as the ``continuous variable \KD distribution'' \autocite{KD_review}.
However, there exist generalized symmetric observables that are \KD positive, and thanks to Theorem \ref{KD positive on second countable LCA groups} we are able to list the generalized pure states that are \KD positive. As the closed subgroups of $\mathbb{R}$ are $\{0\}, \mathbb{R}$ and $a\mathbb{Z}$ for all $a>0$, the KD-positive generalized pure states are given by, up to any complex constant:
\begin{itemize}
    \item The ``position basis'' $(\ket{x})_{x\in \mathbb{R}}$, where $\braket{y|x}=\delta(x-y)$ for any $y\in \mathbb{R}$.
    \item The ``momentum basis'' $(\ket{p})_{p\in \mathbb{R}}$, where $\braket{y|p}=e^{ipy}$ for any $y\in \mathbb{R}$.
    \item The Dirac combs $(\ket{\psi^a_{x,p}})_{a>0,x\in  \mathbb{R},p\in \mathbb{R}}$, where
    \begin{equation*}
        \braket{y|\psi^a_{x,p}}=\sum_{n\in \mathbb{Z}}e^{inp}\delta(y-an +x),\:  \forall y\in \mathbb{R}.
    \end{equation*}
    Those generalized states are also known as ``GKP states'' \autocite{GKP}.
\end{itemize}

Now let's consider the case of the circle $G=\mathbb{S}^1\subset \mathbb{C}$. It is a compact and connected \gls{sclca} group. The closed subgroups of $\mathbb{S}^1$ are the sets of $N$-roots of unity $\{\omega_n=e^{2\pi i \frac{n}{N}}\}_{1\leq n\leq N}$ for any $N\geq 1$, and $\mathbb{S}^1$. The corresponding KD-positive generalized states are, up to any complex constant:
\begin{itemize}
    \item $(\ket{\psi^N_{w,k}})_{N\geq1,w\in \mathbb{S}^1,k\in \mathbb{Z}}$, where
    \begin{equation}
        \braket{z|\psi^N_{w,k}}=\sum_{n=1}^Ne^{-2\pi i\frac{nk}{N}}\delta(z-\omega_n+w),\:\forall z\in \mathbb{S}^1.
    \end{equation}
    \item $\ket{k}_{k\in \mathbb{Z}}$, where $\braket{z|k}=z^k$ for any $z\in \mathbb{S}^1$.
\end{itemize}
Since $\mathbb{S}^1$ is compact, the family $\ket{k}_{k\in \mathbb{Z}}$ of KD-positive generalized pure states consists of square integrable functions. Therefore, in this case, we have $\KDpospure=\{\proj{k}:k\in \mathbb{Z}\}$. Moreover,  we are able to describe the sets $\KDpos$ and $\KDr$ as
\begin{equation}
    \KDr=\overline{\mathrm{span}_{\mathbb{R}}(\KDpospure)} \text{ and }
        \KDpos=\mathrm{conv}(\KDpospure),
\end{equation}
where the former closure is taken with respect to the Hilbert-Schmidt norm, while $\mathrm{conv}(\KDpospure)$ denotes the closed convex hull of $\KDpospure$ in the trace norm topology.
This means that the KD-positive states and KD-real Hilbert-Schmidt observables are exactly those that are diagonal in the Fourier basis. This result holds true whenever the considered \gls{sclca} group $G$ is compact and connected, and it is the last main result of our paper, proved in section \ref{3rd result}. In this case the set $\KDpospure$ consists of all the orthogonal projectors $\proj{\chi}$ with $\chi\in \widehat{G}$ and we have

\begin{thm}
\label{EKD+ on connected}
    Let $G$ be a connected, compact, \gls{sclca} group. Then
    \begin{equation}
    \label{no exotic states}
        \KDr=\overline{\mathrm{span}_{\mathbb{R}}(\KDpospure)} \text{ and }
        \KDpos=\mathrm{conv}(\KDpospure),
    \end{equation}
    where the former closure is taken with respect to the Hilbert-Schmidt topology while the latter closed convex hull is taken in the trace norm topology.
\end{thm}

The problem of characterizing groups for which one or two equalities of \eqref{no exotic states} hold remains open.
We conjecture that
\begin{equation}
    \KDr=\overline{\mathrm{span}_{\mathbb{R}}(\KDpospure)}
\end{equation}
is always true. It is known to hold for any finite abelian group \autocite{KD_finite_groups}.
Regarding the equality
\begin{equation}
\label{conv}
    \KDpos=\mathrm{conv}(\KDpospure),
\end{equation}
it is proved in \autocite{KD_finite_groups} that it is true for the abelian groups $\mathbb{Z}/p^k\mathbb{Z}$, where $p$ is any prime number and $k\geq 1$. However, the authors prove that it does not hold in general for finite abelian groups, as it does not hold for the groups $\mathbb{Z}/2\mathbb{Z}\times \mathbb{Z}/2\mathbb{Z}$ and $\mathbb{Z}/6\mathbb{Z}$.
In view of those results and of the present paper, a lead for thought could be to consider \gls{sclca} groups $G$ for which the set of open and compact subgroups of $G$ is totally ordered for inclusion. We conjecture that this property is sufficient for \eqref{conv} to hold. It might also be necessary.

We conclude in section \ref{KD vs Wigner} by comparing some of our results with earlier results about the positivity of quantum states with respect to the Wigner distribution.

\section{The Kirkwood-Dirac distribution over SCLCA groups}

In this preliminary section, we first collect the mathematical properties of \gls{sclca} groups, their associated Weyl-Heisenberg group and its unitary and irreducible representation on $L^2(G)$ needed in what follows. We then define the characteristic function of any Hilbert-Schmidt operator and show that its symplectic Fourier transform is equal to the \KD distribution. This allows us to make a connection between the \KD distribution and the Wigner distribution defined in \autocite{Crann}. We expect all results in this section to be known to the experts, but we present them here in a form suitable for our purposes.

\subsection{Preliminaries}
\label{Preliminaries}
We begin by recalling some facts about \gls{sclca} groups. Unless stated otherwise, any result mentioned in this subsection can be found in \autocite{Rudin} (chapter I).

A \gls{sclca} group is an abelian group $G$ endowed with a topology for which the group law and inversion map are continuous, and that makes $G$ a Hausdorff, locally compact and second countable space.
The Pontryagin dual $\widehat{G}$ associated with $G$ is the set of continuous group morphisms from $G$ to the unit circle $\mathbb{S}^1\subset \mathbb{C}$. Endowed with pointwise multiplication and the topology of uniform convergence on compact subsets, $\widehat{G}$ is also a \gls{sclca} group if $G$ is a \gls{sclca} group \autocite{Pontryagin} (chapter V), and Pontryagin duality states that there is an isomorphism
\begin{equation}
    \label{Pontryagin duality}
    \widehat{\widehat{G}}\simeq G.
\end{equation}
In this paper, we will use additive notation for the group law of $G$ and multiplicative notation for the group law of $\widehat{G}$.
If $H$ is a closed subgroup of $G$, both $H$ and $G/H$ are \gls{sclca} groups for the natural topologies. The annihilator of a closed subgroup $H\subset G$ is defined as
\begin{equation}
    H^{\perp}=\{\chi\in \widehat{G}:\chi(g)=1,\: \forall g\in H\}.
    \end{equation}
It is a closed subgroup of $\widehat{G}$, and, in view of Pontryagin duality \eqref{Pontryagin duality}, $(H^{\perp})^{\perp}\simeq H$. There is moreover a canonical isomorphism
\begin{align*}
\label{isomorphism}
   H^{\perp} &\to \widehat{G/H}\\
    \chi&\mapsto ([g]\mapsto \chi(g)).
\end{align*}

The general structure of \gls{sclca} groups is known and will be useful to us: any such group is isomorphic to a group of the form
\begin{equation}
    \label{structure theorem}
    \mathbb{R}^n\times H_1,
\end{equation}
where $n\geq 0$ is an integer and $H_1$ is a \gls{sclca} group containing a compact open subgroup. This result is known as the structure theorem \autocite{Weil}.

For a \gls{sclca} group $G$, we fix a Haar measure, denoted $\mu_G$. We recall that a Haar measure is a (positive) Radon measure which is invariant under the translations by any element of $G$. Such a measure always exists and is unique up to a positive scaling. The group $G$ is then compact if and only if $\mu_G$ is finite and in that case, we assume that $\mu_G$ is rescaled so that it is a probability measure. For $p>0$, $L^p(G)$ denotes the usual $p$-Lebesgue space, with respect to the chosen Haar measure $\mu_G$. For any $\psi\in L^1(G)$, the Fourier transform of $\psi$ is the continuous map $\hat{\psi}$ defined on $\widehat{G}$ by
\begin{equation}
    \label{Fourier transform}
    \hat{\psi}(\chi)=\int_G\psi(g)\overline{\chi(g)}d\mu_G(g).
\end{equation}
The Fourier inversion formula states that the Haar measure $\mu_{\widehat{G}}$ on $\widehat{G}$ can be rescaled so that, whenever $\psi\in L^1(G)$ and $\hat{\psi}\in L^1(\widehat{G})$, 
\begin{equation}
    \label{Fourier inversion}
    \psi(g)=\int_{\widehat{G}}\hat{\psi}(\chi)\chi(g)d\mu_{\widehat{G}}(\chi)
\end{equation}
holds almost everywhere. From now on, having fixed a Haar measure on $G$, we assume that the Haar measure on $\widehat{G}$ is scaled so that the Fourier inversion formula holds. In that case, Plancherel's theorem ensures that for any $\psi,\phi\in L^1(G)\cap L^2(G)$
\begin{equation*}
    \int_G\overline{\psi(g)}\phi(g)d\mu_G(g)=\int_{\widehat{G}}\overline{\hat{\psi}(\chi)}\hat{\phi}(\chi)d\mu_{\widehat{G}}(\chi),
\end{equation*}
and thus the Fourier transform extends to a unitary operator $\mathcal{F}:L^2(G)\to L^2(\widehat{G})$.

In \autocite{Bruhat_distributions}, F. Bruhat extends the notion of Schwartz functions on $\mathbb{R}^n$ to arbitrary \gls{lca} groups $G$ via an inductive limit process. We refer to \autocite{Osborne_Schwartz_functions} for a simple characterization of Schwartz-Bruhat functions. In this paper, the set of Schwartz-Bruhat functions will be denoted by $\mathcal{S}(G)$. The space $\mathcal{S}'(G)$ of tempered distributions is then defined as the dual of $\mathcal{S}(G)$ with respect to the inductive limit topology. Three facts will be useful to us. The first one is that, in the sense of distributions
\begin{equation}
    \int_{G}\chi(g)d\mu_G(g)=\delta_1(\chi),\: \forall \chi\in \widehat{G}
\end{equation}
and
\begin{equation}
    \int_{\widehat{G}}\chi(g)d\mu_{\widehat{G}}(\chi)=\delta_0(g),\: \forall g\in G.
\end{equation}
Those are just the Fourier inversion formula over $G$ and $\widehat{G}$, stated in the language of tempered distributions. We also need that the space $\mathcal{S}(G)$ is dense in $L^2(G)$, and that the Fourier transform $\mathcal{F}$ is a continuous bijection between $\mathcal{S}(G)$ and $\mathcal{S}(\widehat{G})$, see \autocite{Osborne_Schwartz_functions}. The Fourier transform therefore extends by duality to a bijective continuous map $\mathcal{F}:\mathcal{S}'(G)\to\mathcal{S}'(\widehat{G})$.

This setting brings a natural Hilbert space on which one can study quantum mechanics: $L^2(G)$. We will denote by $\mathcal{T}_1(G)$ the space of trace class operators of $L^2(G)$ and $\mathcal{T}_2(G)$ the space of Hilbert-Schmidt operators. The former is a Banach space for the trace norm $||A||_1=\tr(|A|)$, while the latter is a Hilbert space for the inner product $\langle A,B\rangle=\tr(A^*B)$. It is well known that $\mathcal{T}_2(G)$ is unitarily isomorphic to $L^2(G\times G)$, via $(A\psi)(g)=\int_Gk_{A}(g,g')\psi(g')d\mu_G(g')$, where $k_{A}\in L^2(G\times G)$ is called the integral kernel representation of $A$ \autocite{Reed_and_Simon}. Note that we have $\mathcal{T}_1(G)\subset \mathcal{T}_2(G)$ as sets. 

Second countability of $G$ makes $L^2(G)$ a separable Hilbert space (see for example \autocite{Cohn}). This then makes $\mathcal{T}_1(G)$ a separable Banach space, and thus a separable dual Banach space, as it is the dual of the Banach space of compact operators on $L^2(G)$ endowed with the operator norm \autocite{Reed_and_Simon}.

A state on $L^2(G)$ is an operator $\rho\in \mathcal{T}_1(G)$ which is positive semi-definite and such that $\tr(\rho)=1$. If a state $\rho$ is a rank one orthogonal projector to the subspace generated by some $\psi\in L^2(G)$ with $||\psi||_2=1$, it is called a pure state, and it is denoted in Dirac notation $\rho=\proj{\psi}$. The integral kernel of such a pure state is given by $k_{\proj{\psi}}(g,g')=\psi(g)\overline{\psi(g')}$. The set of states is a closed convex subset of $\mathcal{T}_1(G)$, whose extreme points are exactly the pure states. A state which is not pure is called a mixed state.

An observable on $L^2(G)$ is a bounded self-adjoint operator. A generalized operator $A$ is, as mentioned in definition \ref{generalized operator}, a linear operator represented by a distributional kernel $k_A\in \mathcal{S}'(G\times G)$.

We finish by mentioning that while all the definitions that we will give carry over naturally to arbitrary \gls{lca} groups, the proof of Theorem \ref{KD positive on second countable LCA groups} and everything that follows require the assumption of second countability. 

\subsection{The Kirkwood-Dirac distribution}

We recall from definition \ref{KD def} that, for a generalized operator $A$ with distributional kernel $k_A$, its \KD distribution is the tempered distribution on $G\times \widehat{G}$
\begin{equation*}
    \mathrm{KD}_{A}(g,\chi)=\overline{\chi(g)}\int_{G}k_{A}(g,g')\chi(g')d\mu_G(g'),
\end{equation*}
One observes that $\mathrm{KD}_{A}$ is nothing more than a character multiplied by the (inverse) Fourier transform of $k_{A}$ with respect to the second variable. 
It follows directly from the Plancherel formula that, restricted to the space of Hilbert-Schmidt operators on $L^2(G)$,
\begin{equation}
\label{KD unitary}
    \mathrm{KD}:\mathcal{T}_2(G)\to L^2(G\times \widehat{G})
\end{equation}
is a unitary isomorphism, in particular it is invertible. The integral kernel of the operator is recovered from the \KD distribution by the formula
\begin{equation}
    \label{inverse KD}
    k_A(g,g')=\int_{\widehat{G}}\mathrm{KD}_{A}(g,\chi)\chi(g-g')d\mu_{\widehat{G}}(\chi).
\end{equation}
The fact that it is unitary also implies that the associated symbol is the \KD distribution itself, as for any $A,B\in \mathcal{T}_2(G)$,
\begin{equation}
    \tr(A^*B)=\int_{G\times \widehat{G}}\overline{\mathrm{KD}_A(g,\chi)}\mathrm{KD}_B(g,\chi)d\mu_G(g)d\mu_{\widehat{G}}(\chi).
\end{equation}
This fact is mathematically very convenient and it is at the core of the \KD representation. Indeed, it is an overlap formula (see \eqref{overlap}) and thus defines a quasiprobability representation.

 If $\rho$ is a quantum state such that $k_{\rho}\in \mathcal{S}(G\times G)$, we can compute
\begin{align}
\label{marginals}
    &\int_{\widehat{G}}\mathrm{KD}_{\rho}(g,\chi)d\mu_{\widehat{G}}(\chi)=k_{\rho}(g,g):=\bra{g}\rho\ket{g},\\
    &\int_{G}\mathrm{KD}_{\rho}(g,\chi)d\mu_G(g)=\int_{G^2}k_{\rho}(g,g')\overline{\chi(g)}\chi(g')d\mu_G(g)d\mu_G(g'):=\bra{\chi}\rho\ket{\chi}\\
    &\int_{G\times \widehat{G}}\mathrm{KD}_{\rho}d\mu_Gd\mu_{\widehat{G}}=\tr(\rho)=1.
\end{align}
Those three expressions show that the \KD distribution of a quantum state is a quasiprobability distribution that is compatible with the Born rule: $\bra{g}\rho\ket{g}$ is usually interpreted as the probability distribution of the position of the particle described by the quantum system $L^2(G)$, while $\bra{\chi}\rho\ket{\chi}$ is interpreted as the probability distribution of the momentum of the particle. This makes the \KD distribution a quantum equivalent of a joint probability distribution.
However, $\mathrm{KD}_{\rho}$ is in general not a probability measure on $G\times \widehat{G}$, as it may take negative or even non-real values. 

\subsection{Weyl-Heisenberg group, Kirkwood-Dirac and Wigner-Weyl distributions and ordered quantizations of phase space}
\label{Ordering}

The groups $G$ and $\widehat{G}$ act naturally on $L^2(G)$ in the following fashion:
For $g\in G$ and $\chi\in \widehat{G}$, define the operators $T_g$ and $M_{\chi}$ by
\begin{equation}
    (T_g\psi)(g')=\psi(g'-g)
\end{equation}
and
\begin{equation}
    (M_{\chi}\psi)(g')=\chi(g')\psi(g'),
\end{equation}
for any $\psi\in L^2(G)$ and any $g'\in G$.
It follows from invariance of the Haar measure that both those operators are unitary. The set of unitary operators on $L^2(G)$ is denoted $\mathcal{U}(L^2(G))$. Moreover, they satisfy the commutation relations
\begin{equation*}
    M_{\chi}T_g=\chi(g)T_gM_{\chi}
\end{equation*}
for any $(g,\chi)\in G\times \widehat{G}$. To keep track of this phase factor, one defines the Weyl-Heisenberg group associated with $G$ by $H(G)=G\times\widehat{G}\times \mathbb{S}^1$, with group law
\begin{equation}
\label{group law}
    (g,\chi,z)\cdot(g',\chi',z')=(g+g',\chi\chi',zz'\overline{\chi'(g)})
\end{equation}
This group then acts naturally on $L^2(G)$ with the following properties:
\begin{thm}
\label{Weyl-Heisenberg}
    The following action of $H(G)$ over $L^2(G)$ is unitary and irreducible:
    \begin{align*}
        U:H(G)&\longrightarrow \mathcal{U}(L^2(G))\\
        (g,\chi,z)&\mapsto zM_{\chi}T_g.
    \end{align*}
\end{thm}
The proof of this theorem can be found in \autocite{Stone_VonNeumann_Mackey}. This action encodes the translations in position space ($G$) and in momentum space ($\hat{G}$).
It extends to an action on $\mathcal{S}'(G)$, by the duality
\begin{equation*}
    \langle U(g,\chi,z)\psi,\phi\rangle=\langle \psi,U(-g,\overline{\chi},\overline{z})\phi\rangle
\end{equation*}
for $\psi\in \mathcal{S}'(G)$ and $\phi\in \mathcal{S}(G)$.
From \eqref{group law}, we have
\begin{equation}
    \label{unitary inverse}
    U(g,\chi,z)^{-1}=U(g,\chi,z)^*=U(-g,\overline{\chi},\overline{z\chi(g)}).
\end{equation}
The action of the Weyl-Heisenberg group  acts as a translation on the \KD distribution, in particular it preserves properties such has positiveness. Indeed for any $\psi\in \mathcal{S}'(G)$ and any $(g_0,\chi_0,z)\in H(G)$,
\begin{equation}
\label{KD positivity is preserved by W-H}
    \mathrm{KD}_{U(g_0,\chi_0,z)\psi}(g,\chi)=\mathrm{KD}_{\psi}(g-g_0,\chi\overline{\chi_0}),
\end{equation}
where this equation is to be understood in the sense of distributions. Similarly for any generalized operator $A$
\begin{equation}
    \mathrm{KD}_{U(g_0,\chi_0,z)A U(g_0,\chi_0,z)^{*}}(g,\chi)=\mathrm{KD}_{A}(g-g_0,\chi\overline{\chi_0}).
\end{equation}

For mathematical simplicity and until the end of this section, we restrict ourselves to Hilbert-Schmidt operators, and we will come back to generalized operators in the next section.
Inspired by the common methods employed on $\mathbb{R}^n$, we introduce the characteristic function of any Hilbert-Schmidt operator $A$:

\begin{definition}
    Let $A$ be a Hilbert-Schmidt operator on $L^2(G)$. The \Def{characteristic function} of $A$ is the function defined on $G\times \widehat{G}$ by
    \begin{equation}
    \label{def characteristic}
        \mathfrak{X}^0_{A}(g,\chi)=\tr(AU(g,\chi,1)).
    \end{equation}
\end{definition}

Note that \eqref{def characteristic} is a priori well-defined only for trace class operators. However, the following lemma shows that $\mathfrak{X}^0$ extends to a unitary map between $\mathcal{T}_2(G)$ and $\mathrm{Ran}(\mathfrak{X}^0)\subset L^2(G\times G)$:

\begin{lemma}
    Let $A,B\in \mathcal{T}_2(G)$ such that $k_A,k_B\in \mathcal{S}(G\times G)$. Then
    \begin{equation*}
        \tr(A^*B)=\langle \mathfrak{X}^0_A,\mathfrak{X}^0_B\rangle_{G\times \widehat{G}}.
    \end{equation*}
    As a consequence, $\mathfrak{X}^0$ is a well-defined unitary map between $\mathcal{T}_2(G)$ and $\mathrm{Ran}(\mathfrak{X}^0)$.
\end{lemma}

\begin{proof}
    It is easily checked that the kernel of the operator $AU(g,\chi)$ is given by $(g',g'')\mapsto \chi(g''+g)k_A(g',g''+g)$ and similarly for $B$. We compute
    \begin{align*}
        &\langle \mathfrak{X}^0_A,\mathfrak{X}^0_B\rangle_{G\times \widehat{G}}\\
        =&\int_{G\times \widehat{G}}\overline{\tr(AU(g,\chi,1))}\tr(BU(g,\chi,1))d\mu_G(g)d\mu_{\widehat{G}}(\chi)\\
        =&\int_{G^3\times \widehat{G}}\overline{k_A(g',g'+g)}k_B(g'',g''+g)\chi(g''-g')d\mu_{\widehat{G}}(\chi)d\mu_G(g)d\mu_G(g')d\mu_G(g'')\\
        =&\int_{G^2}\overline{k_A(g',g'+g)}k_B(g',g'+g)d\mu_G(g)d\mu_G(g')\\
        =&\int_{G^2}\overline{k_A(g',g)}k_B(g',g)d\mu_G(g)d\mu_G(g')\\
        =&\tr(A^*B).
    \end{align*}
    By density of $\mathcal{S}(G\times G)$ in $L^2(G\times G)$, $\mathfrak{X}^0$ can indeed be extended to a unitary map between $\mathcal{T}_2(G)$ and $\mathrm{Ran}(\mathfrak{X}^0)$.
\end{proof}
One could have considered the action of the Weyl-Heisenberg group given by the operators $T_gM_{\chi}=\overline{\chi(g)}M_{\chi}T_g$. 
It is therefore also natural to consider the function 
\begin{equation}
    \mathfrak{X}^1_A(g,\chi)=\tr\left( AT_gM_{\chi}\right)=\mathfrak{X}^0_A(g,\chi)\overline{\chi(g)}.
\end{equation}
Usually over $\mathbb{R}$, the characteristic function of the Hilbert-Schmidt operator $A$ is defined as
\begin{equation*}
    (x,p)\mapsto\tr(Ae^{i(pX-xP)}),
\end{equation*}
where $X$ and $P$ are position and momentum operators \autocite{Cahill_Glauber}. Since $[X,P]=i$, we have from the Baker-Campbell-Hausdorff formula $e^{i(pX-xP)}=e^{ipX}e^{-ixP}e^{-i\frac{px}{2}}$. In our notation, this corresponds to the unitary operator $U(x,p,e^{-i\frac{xp}{2}})$. There is no equivalent in the general setting. However, if we assume that the doubling map $g\mapsto g+g$ is invertible, and that we denote it's inverse by $g\mapsto \frac{g}{2}$, we can define for any Hilbert-Schmidt operator $A$
\begin{equation}
    \mathfrak{X}^{\frac{1}{2}}_A(g,\chi)=\tr\left(A U\left(g,\chi,\overline{\chi\left(\frac{g}{2}\right)}\right)\right)=\mathfrak{X}^0_A(g,\chi)\overline{\chi\left(\frac{g}{2}\right)}.
\end{equation}
When it is well-defined, $\mathfrak{X}^{\frac{1}{2}}$ can be thought of as an interpolation between $\mathfrak{X}^0$ and $\mathfrak{X}^1$. Moreover, the Wigner-Weyl distribution $\mathrm{W}_A$ of the Hilbert-Schmidt operator $A$ is then defined has the symplectic Fourier transform of $\mathfrak{X}^{\frac{1}{2}}_A$, see \autocite{Crann}. We recall the definition of the symplectic Fourier transform:

\begin{definition}
    Let $G$ be a \gls{sclca} group. The \Def{symplectic Fourier transform} is the unitary map $\mathcal{F}_{\mathrm{symp}}:L^2(G\times \widehat{G})\to L^2(G\times \widehat{G})$ defined by
    \begin{equation}
        \mathcal{F}_{\mathrm{symp}}(f)(g,\chi)=\int_{G\times \widehat{G}}f(g',\chi')\chi(g')\overline{\chi'(g)}d\mu_G(g')d\mu_{\widehat{G}}(\chi'),
    \end{equation}
    for any $f\in L^2(G\times \widehat{G})$.
\end{definition}

Over $\mathbb{R}$, it was shown in \autocite{OConell} that taking the symplectic Fourier transform of $\mathfrak{X}^1_A$ returns the \KD distribution of $A$. We prove the following generalization to \gls{sclca} groups, needed in section \ref{Identifying classical fragment}.

\begin{thm}[O'Conell formula]
    Let $A$ be a Hilbert-Schmidt operator on $L^2(G)$. Its \KD distribution and characteristic function are related as follows:
    \begin{equation}
    \label{KD and characteristic}
        \mathrm{KD}_{A}=\mathcal{F}_{\mathrm{symp}}(\mathfrak{X}_{A}^1).
    \end{equation}
\end{thm}

\begin{proof}
    Assume as before that $k_{A}\in \mathcal{S}(G\times G)$. Then
    \begin{align*}
        \mathfrak{X}^0_{A}(g',\chi')=\tr(A U(g',\chi'))=\int_Gk_{A}(g'',g''+g')\chi'(g'')\chi'(g')d\mu_G(g'').
    \end{align*}
    We compute
    \begin{align*}
        &\mathcal{F}_{\mathrm{symp}}(\mathfrak{X}^1_A)(g,\chi)\\
        =&\int_{G\times \widehat{G}}\mathfrak{X}^0_{A}(g',\chi')\overline{\chi'(g')}\chi(g')\overline{\chi'(g)}d\mu_G(g')d\mu_{\widehat{G}}(\chi')\\
        =&\int_{G^2\times \widehat{G}}k_{A}(g'',g''+g')\chi'(g'')\chi(g')\overline{\chi'(g)}d\mu_G(g')d\mu_{\widehat{G}}(\chi')d\mu_G(g'')\\
        =&\int_{G^2}k_{A}(g'',g''+g')\delta(g-g'')\chi(g')d\mu_G(g')d\mu_G(g'')\\
        =&\int_Gk_{A}(g,g+g')\chi(g')d\mu_G(g')\\
        =&\int_Gk_{A}(g,g')\chi(g'-g)d\mu_G(g')\\
        =&\mathrm{KD}_{A}(g,\chi),
    \end{align*}
    where we used several times the invariance of the Haar measure. The right and left hand sides of \eqref{KD and characteristic} are unitary maps in the variable $A$, so the fact that it holds if $k_{A}\in \mathcal{S}(G\times G)$ implies by density that it holds for any $A\in \mathcal{T}_2(G)$.
\end{proof}

\begin{rmk}
    It also follows from O'Conell's formula that $\mathrm{Ran}(\mathfrak{X}^0)=L^2(G\times G)$. Indeed,
    \begin{equation*}
        \mathfrak{X}^0_{A}(g,\chi)=\chi(g)\times\mathcal{F}_{\mathrm{symp}}^{-1}(\mathrm{KD}_{A})(g,\chi),
    \end{equation*}
    and the right hand side is surjective from $\mathcal{T}_2(G)$ to $L^2(G\times G)$, as a function of $A$.
\end{rmk}

It is easy to see that
\begin{equation*}
    \overline{\mathrm{KD}_{A^*}}=\overline{\mathcal{F}_{\mathrm{symp}}(\mathfrak{X}^1_{A^*})}=\mathcal{F}_{\mathrm{symp}}\left(\overline{(\mathfrak{X}^1_{A^*})^-}\right),
\end{equation*}
where for all $(g,\chi)\in G\times\widehat{G}$
\begin{equation*}
    \overline{(\mathfrak{X}^1_{A^*})^-}(g,\chi)=\overline{(\mathfrak{X}^1_{A^*})(-g,\overline{\chi})}=\chi(g)\tr(A U(-g,\overline{\chi},1)^*)=\mathfrak{X}^0_A(g,\chi)
\end{equation*}
where we used \eqref{unitary inverse}. Therefore 
\begin{equation}
    \label{KD conjugate}
    \overline{\mathrm{KD}_{A^*}}=\mathcal{F}_{\mathrm{symp}}(\mathfrak{X}^0_A):=\mathrm{aKD}_A
\end{equation}
for any Hilbert-Schmidt operator $A$. We will call the distribution $\mathrm{aKD}$ the anti-KD distribution. Note that if $A$ is self-adjoint, then $\mathrm{aKD}_A=\overline{\mathrm{KD}_A}$. Altogether, we have
\begin{align*}
    \mathrm{KD}_A&=\mathcal{F}_{\mathrm{symp}}(\mathfrak{X}^1_A),\\
    \mathrm{W}_A&=\mathcal{F}_{\mathrm{symp}}(\mathfrak{X}^{\frac{1}{2}}_A), \text{ when it exists,}\\
    \mathrm{aKD}_A&=\mathcal{F}_{\mathrm{symp}}(\mathfrak{X}^0_A).
\end{align*}
Thus the Wigner function, when it exists, can be interpreted as an interpolation between the \KD distribution and the anti-KD distribution.

 As already mentioned in \eqref{KD unitary}, the \KD distribution is a unitary map between the Hilbert-Schmidt operators of $L^2(G)$ and the space $L^2(G\times \widehat{G})$ and its inverse is given by \eqref{inverse KD}. Let $f\in L^2(G)$ and $h\in L^2(\widehat{G})$. Consider the operators $\mathbf{f}_G$ and $\mathbf{h}_{\widehat{G}}$ on $L^2(G)$ defined by
\begin{align*}
    &(\mathbf{f}_G\psi)(g)=f(g)\psi(g)\\
    &\left(\mathcal{F}\mathbf{h}_{\widehat{G}}\mathcal{F}^{-1}\hat{\psi}\right)(\chi)=h(\chi)\hat{\psi}(\chi),
\end{align*}
So that $\mathbf{f}_G$ is the multiplication operator by $f$ and the conjugation of $\mathbf{h}_{\widehat{G}}$ by the Fourier transform is the multiplication operator by $h$.
It is easily checked that the operator
\begin{equation}
    \mathbf{f}_G\mathbf{h}_{\widehat{G}}
\end{equation}
is a Hilbert-Schmidt operator on $L^2(G)$ whose kernel is given by
\begin{equation*}
    k(g,g')=f(g)(\mathcal{F}^{-1}h)(g-g').
\end{equation*}
It follows
\begin{align*}
    \mathrm{KD}_{\mathbf{f}_G\mathbf{h}_{\widehat{G}}}(g,\chi)&=\overline{\chi(g)}\int_Gk(g,g')\chi(g')d\mu_G(g')\\
    &=\int_Gf(g)(\mathcal{F}^{-1}h)(g-g')\overline{\chi(g-g')}d\mu_G(g')\\
    &=f(g)h(\chi)
\end{align*}
for any $(g,\chi)\in G\times \widehat{G}$, and thus
\begin{equation}
    \mathrm{KD}^{-1}(f\otimes h)=\mathbf{f}_G\mathbf{h}_{\widehat{G}}.
\end{equation}
This expression matches the definition of the Kohn-Nirenberg quantization on \gls{sclca} groups given in \autocite{Kohn_Nirenberg_quantization} and in \autocite{FeichtingerKozek98}.
Note that we have
\begin{equation}
    f(g)h(\chi)=h(\chi)f(g)
\end{equation}
for all $(g,\chi)\in G\times \widehat{G}$, while
\begin{equation}
    \mathbf{f}_G\mathbf{h}_{\widehat{G}}\neq \mathbf{h}_{\widehat{G}} \mathbf{f}_G
\end{equation}
in general.
The inverse of the Kirkwood-Dirac distribution therefore gives a choice of ordering. It quantizes functions on phase space by placing functions of the variable $g$ ``to the left'' and functions of the variable $\chi$ ``to the right''. When $G=\mathbb{R}^n$, this corresponds to the standard ordering, when functions of the position operators $X_i$ are to the left while functions of the momentum operator $P_i$ are to the right. Taking the inverse of the anti-KD distribution would yield the anti-standard ordering instead.

Still in the case $G=\mathbb{R}^n$, it is well known that the Weyl quantization (the inverse of the Wigner function), acts as a completely symmetric quantization, mapping for example (when $n=1$) the function $(x,p)\mapsto xp$ to the self-adjoint operator $\frac{1}{2}(XP+PX)$. It interpolates between the normal ordering (annihilation operator to the left), which is related to the Glauber-Sudarshan $P$ function, and anti-normal ordering (creation operator to the left), which is related to the Husimi $Q$ function \autocite{Cahill_Glauber}. It also interpolates between the standard ordering, which corresponds to the \KD distribution and anti-standard ordering, which corresponds to the anti-KD distribution \autocite{OConell_Wang}.

Summing up the above discussion, for any \gls{sclca} group $G$ having an invertible doubling map, the Wigner-Weyl distibution is well-defined and interpolates between the \KD distribution and the anti-KD distribution. In that sense, the inverse of the Wigner-Weyl distribution defined in \autocite{Crann} can be viewed as a symmetric quantization of the phase space $G\times \widehat{G}$, even if the position and momentum operators don't necessarily exist.

\section{Characterization of the Kirkwood-Dirac positive generalized pure states on SCLCA groups}
\label{Main result}

In this section we provide a complete characterization of the generalized pure states on $G$ that have a positive \KD distribution. We begin by constructing a family of KD-positive generalized pure states using the Poisson-Tate formula associated with closed subgroups and the action of the Weyl-Heisenberg group. We then show that this family in fact contains all the KD-positive generalized pure states. This yields the proof of Theorem \ref{KD positive on second countable LCA groups}.

\subsection{Poisson-Tate formula}

In this subsection, we construct a family of generalized pure states having a positive \KD distribution.

The form of equation \eqref{KD pure} implies that $\mathrm{KD}_{\psi}$ is a positive measure only if both $\psi$ and $\hat{\psi}$ are a minima tempered complex measures. We recall that a tempered measure is a measure which is also a tempered distribution. This motivates the following definition:

\begin{definition}
    Let $G$ be a \gls{sclca} group. A \Def{Fourier measure} on $G$ is a tempered complex measure $\psi$ such that $\hat{\psi}$ is also a tempered complex measure.
\end{definition}
As examples of Fourier measures, one finds the square integrable functions and the finite measures. A family of Fourier measures that will be very interesting for us is given by the generalized Poisson formula, or Poisson-Tate formula, see for example \autocite{Poisson_formula}:

\begin{thm}[Poisson-Tate formula]
Let $G$ be a \gls{lca} group, $H\subset G$ be a closed subgroup and $\mu_H$ be a Haar measure on $H$. Then $\mu_H$ is a Fourier measure on $G$ and its Fourier transform $\hat{\mu}_H$ is a Haar measure on $H^{\perp}$. In other words, there exists a scaling of the Haar measure $\mu_{H^{\perp}}$ on $H^{\perp}$ such that for any $f\in \mathcal{S}(G)$, we have
\begin{equation}
    \int_Hfd\mu_H=\int_{H^{\perp}}\hat{f}d\mu_{H^{\perp}}.
\end{equation}
\end{thm}

For $H\subset G$ a closed subgroup, we can then explicitly compute the \KD distribution of a Haar measure $\mu_H$:
\begin{equation}
    \mathrm{KD}_{\mu_H}(g,\chi)=\overline{\chi(g)}\mu_H(g)\overline{\hat{\mu}_H(\chi)}=\mu_H(g)\mu_{H^{\perp}}(\chi),
\end{equation}
which is a positive measure on $G\times \widehat{G}$, supported by $H\times H^{\perp}$. This observation combined with the fact that the Weyl-Heisenberg group preserves KD-positivity implies that the generalized pure states associated with the tempered distributions
\begin{equation}
\label{translates of Haar measures}
    \psi^H_{g,\chi,z}=U(g,\chi,z)\mu_H
\end{equation}
are all KD positive, where $(g,\chi,z)\in H(G)$ and $\mu_H$ is a Haar measure on some closed subgroup $H\subset G$. 

\subsection{Proof of Theorem \ref{KD positive on second countable LCA groups}}
The previous subsection proves the direct implication in Theorem \ref{KD positive on second countable LCA groups}. We now prove the indirect implication. For clarity, we state what remains to be proven:

\begin{prop}
\label{indirect implication}
    Let $G$ be a \gls{sclca} group and $\psi$ be a non-zero tempered measure such that the associated generalized pure state $\proj{\psi}$ has a positive \KD distribution. Then there exists a closed subgroup $H\subset G$ such that $\psi$ is, up to the action of the Weyl-Heisenberg group, a Haar measure on $H$.
\end{prop}

The final ingredient of the proof is the following lemma, which can be seen as a weak multiplicative uncertainty principle, as for example it is a consequence of the Donoho-Stark uncertainty principle \autocite{Donoho_Stark} in the special case of finite abelian groups.

\begin{lemma}[Weak uncertainty principle]
\label{uncertainty principle}
    Let $G$ be a \gls{lca} group and $H\subset G$ be a closed subgroup. Let $\psi$ be a non-zero positive Fourier measure on $G$ such that $\mathrm{supp}(\psi)\subset H$ and $\mathrm{supp}(\hat{\psi})\subset H^{\perp}$. Then $\psi$ is a Haar measure on $H$.
\end{lemma}

Here $\mathrm{supp}$ refers to the usual notion of the support of a positive measure, defined as the complement of the union of open sets of zero measure. The support of a complex measure is by definition the support of its absolute value.

This lemma seems to be well-known in various forms \autocite{Rudin,Donoho_Stark,qualitative_uncertainty_principle, uncertainty_principle_entropy}. The above statement is adapted to our needs. We provide a simple proof for completeness:

\begin{proof}
    Let $h\in H$, $f\in \mathcal{S}(G)$ and consider $f_h(g)=f(g+h)$. Then $f_h\in \mathcal{S}(G)$ and $\hat{f_h}(\chi)=\chi(h)\hat{f}(\chi)$. It follows
    \begin{align*}
        \int_Hf(g+h)d\psi(g)&=\int_Gf_h(g)d\psi(g) \text{ \footnotesize since $\psi$ is supported on $H$}\\
        &=\int_{\widehat{G}}\chi(h)\hat{f}(\chi)d\hat{\psi}(\chi) \text{ \footnotesize by duality}\\
       &= \int_{\widehat{G}}\hat{f}(\chi)d\hat{\psi}(\chi) \text{ \footnotesize since $\hat{\psi}$ is supported on $H^{\perp}$ }\\
       &=\int_Gf(g)d\psi(g) \text{ \footnotesize by duality}\\
       &=\int_Hf(g)d\psi(g).
    \end{align*}
    Since this holds for arbitrary $f\in 
 \mathcal{S}(G)$, $\psi$ is invariant under translation by any $h\in H$, and is therefore a Haar measure on $H$ ($\psi$ is by assumption Radon since it is a tempered measure).
\end{proof}

\begin{proof}[Proof of Proposition \ref{indirect implication}]
    We recall that the form of equation \eqref{KD pure} implies that $\mathrm{KD}_{\psi}$ is a measure on $G\times \widehat{G}$ only if both $\psi$ and $\hat{\psi}$ are also (complex) measures.
    In the first part of the proof, we show that $\hat{\psi}$ is in fact, up to a constant phase, a positive measure on $\widehat{G}$.  
    
    Since $G$ is second countable, it is metrizable with a proper metric; so is $\widehat{G}$ (see \autocite{Proper_metric}). In this proof, we fix one such metric for $G$ and for $\widehat{G}$.
    For $R>0$, consider $B_{R}\subset \widehat{G}$ the closed ball of $\widehat{G}$, centered at $1$ and of radius $R$, which is meant to become large. Since the metric on $\widehat{G}$ is proper, $B_R$ is compact. Let $A\subset G$ be a compact neighbourhood of $0$. The map 
    \begin{align*}
        A\times B_R&\to\mathbb{R}\\
        (g,\chi)&\mapsto |\chi(g)-1|
    \end{align*}
    is continuous. Since both $A$ and $B_R$ are compact metric spaces, it follows that
    \begin{align*}
        A&\to \mathbb{R}\\
        g&\mapsto \sup_{\chi\in B_R}|\chi(g)-1|
    \end{align*}
        is also continuous on $A$. Then, for any $n\geq 1$, there exists $\delta_n>0$ such that
        \begin{equation*}
            \sup_{\chi\in B_R}|\chi(g)-1|\leq \frac{1}{n}
        \end{equation*}
        whenever $g\in B(0,\delta_n)\cap A$, the open ball $B(0,\delta_n)$ being defined using the proper metric of $G$. Since $A$ is a neigbourhood of $0$ we can assume, up to taking smaller $\delta_n$'s, that $B(0,\delta_n)\subset A$ for all $n\geq 1$.
        
        Now, up to a translation, we can assume that $0\in \mathrm{supp}(\psi)$. Moreover from the Radon-Nikodym theorem, there exists a measurable real valued function $\theta$ on $G$ such that $d\psi=e^{i\theta}d|\psi|=(\cos(\theta)+i\sin(\theta))d|\psi|$. For each $n\geq 1$, define the following sets:
        \begin{align*}
            &B_n^1=\{g\in B(0,\delta_n):\:\cos(\theta(g))\geq \frac{1}{\sqrt{2}}\},\\
            &B_n^2=\{g\in B(0,\delta_n):\:\cos(\theta(g))\leq\frac{-1}{\sqrt{2}}\},\\
            &B_n^3=\{g\in B(0,\delta_n):\:\sin(\theta(g))> \frac{1}{\sqrt{2}}\},\\
            &B_n^4=\{g\in B(0,\delta_n):\:\sin(\theta(g))<\frac{-1}{\sqrt{2}}\}.
        \end{align*}
        Those sets are measurable, two by two disjoint, and
        \begin{equation*}
            |\psi|(B(0,\delta_n))=|\psi|(B_n^1)+|\psi|(B_n^2)+|\psi|(B_n^3)+|\psi|(B_n^4).
        \end{equation*}
        Since $0\in \mathrm{supp}(\psi)=\mathrm{supp}(|\psi|)$, $|\psi|(B(0,\delta_n))>0$, and so there exists $i_0\in\{1,2,3,4\}$ such that $|\psi|(B^{i_0}_n)>0$. A quick distinction of cases shows that
        \begin{equation*}
            |\psi(B_n^{i_0})|=\sqrt{\left(\int_{B_n^{i_0}}\cos(\theta)d|\psi|\right)^2+ \left(\int_{B_n^{i_0}}\sin(\theta)d|\psi|\right)^2} \geq \frac{1}{\sqrt{2}}|\psi|(B^n_{i_0}).
        \end{equation*}
        Define $A_n:=B_n^{i_0}\subset B(0,\delta_n)\subset A$ for all $n\geq 1$.
        We then have for any measurable set $C\subset B_R$
        \begin{align*}
            \left| \frac{1}{|\psi(A_n)|}\int_{A_n\times C}(\overline{\chi(g)}-1)d\psi(g)\overline{d\hat{\psi}(\chi)}\right|&\leq \frac{1}{|\psi(A_n)|}\int_{A_n\times C}|\chi(g)-1|d|\psi|(g)d|\hat{\psi}|(\chi)\\
            &\leq \frac{1}{|\psi(A_n)|}\frac{1}{n}|\psi|(A_n)|\hat{\psi}|(C)\\
            &\leq \frac{\sqrt{2}}{n}|\hat{\psi}|(C)\\
            &\xrightarrow[n\to+\infty]{}0.
        \end{align*}
        Moreover, the sequence $\left(\frac{\psi(A_n)}{|\psi(A_n)|}\right)_{n\geq 1}$ having modulus $1$, we can assume, up to taking a subsequence, that there exists $\lambda_{R}\in \mathbb{S}^1$ such that $\frac{\psi(A_n)}{|\psi(A_n)|}\xrightarrow[n\to +\infty]{}\overline{\lambda_R}$. We stress that, a priori, the value of this constant depends on the chosen
        radius $R$. Then for any measurable set $C\subset B_R$
        \begin{align*}
            0\leq &\frac{1}{|\psi(A_n)|}\mathrm{KD}_{\psi}(A_n\times C)\\
            =&\frac{1}{|\psi(A_n)|}\int_{A_n\times C}\overline{\chi(g)}d\psi(g)d\overline{\hat{\psi}(\chi)}\\
            =&\frac{1}{|\psi(A_n)|}\psi(A_n)\overline{\hat{\psi}(C)} + \frac{1}{|\psi(A_n)|}\int_{A_n\times C}(\overline{\chi(g)}-1)d\psi(g)d\overline{\hat{\psi}(\chi)}\\
            \xrightarrow[n\to+\infty]{}&\overline{\lambda_{R}} \overline{\hat{\psi}(C)}.
        \end{align*}
        Therefore, for any $R>0$, there exists $\lambda_R\in \mathbb{S}^1$ such that $\lambda_R\hat{\psi}\vert_{B_R}$ is a positive measure. Moreover, by definition for any $C\subset B_R$
        \begin{equation*}
            |\hat{\psi}|(C)=\sup_{\cup_nE_n=C}\sum_{n}|\hat{\psi}(E_n)|,
        \end{equation*}
        where the supremum is taken over all countable partitions $(E_n)_n$ of $C$. For such a partition, we have $|\hat{\psi}(E_n)|=\lambda_R\hat{\psi}(E_n)$, hence
        \begin{equation*}
            |\hat{\psi}|(C)=\lambda_R\hat{\psi}(C)=|\hat{\psi}(C)|.
        \end{equation*}
        Since $\hat{\psi}\neq 0$, there exists $R_0$ big enough so that $|\hat{\psi}(B_{R_0})|=|\hat{\psi}|(B_{R_0})>0$.
        For any $R>R_0$, we have
        \begin{equation*}
            \lambda_R\hat{\psi}(B_{R_0})=\lambda_R\hat{\psi}\vert_{B_R}(B_{R_0})=|\hat{\psi}(B_{R_0})|=\lambda_{R_0}\hat{\psi}\vert_{B_{R_0}}(B_{R_0})=\lambda_{R_0}\hat{\psi}(B_{R_0}),
        \end{equation*}
        hence $\lambda_R=\lambda_{R_0}$ for any $R>R_0$. To finish, observe that for all measurable $B\subset \widehat{G}$, we have
        \begin{equation*}
            \hat{\psi}(B)=\lim_{R\to+\infty}\hat{\psi}\vert_{B_R}(B).
        \end{equation*}
        It follows that
        \begin{equation*}
            \lambda_{R_0}\hat{\psi}(B)=\lambda_{R_0}\lim_{R\to+\infty}\hat{\psi}\vert_{B_R}(B)=\lim_{R\to+\infty}\lambda_R\hat{\psi}\vert_{B_R}(B)\geq 0.
        \end{equation*}
        Hence up to a multiplication by a constant, $\hat{\psi}$ is a positive measure.

        We can now conclude the proof. We consider from now on that $\hat{\psi}$ is a positive measure. Up to a translation, we can assume further that $1\in \mathrm{supp}(\hat{\psi})$. For any Borel subset $A\subset G\times \widehat{G}$, we have
    \begin{align*}
        0&\leq \mathrm{KD}_{\psi}(A)\\
        &=\int_{A}\overline{\chi(g)}e^{i\theta(g)}d|\psi|(g)d\hat{\psi}(\chi).
    \end{align*}
     Since $\hat{\psi}$ is a positive measure, we have
    \begin{equation*}
        \overline{\chi(g)}e^{i\theta(g)}\geq 0
    \end{equation*}
    for $\hat{\psi}$-almost every $\chi\in \widehat{G}$ and $|\psi|$-almost every $g\in G$. Since this quantity also has modulus $1$ almost everywhere, we deduce
    \begin{equation}
    \label{equation phase of psi}
        \overline{\chi(g)}e^{i\theta(g)}=1
    \end{equation}
    for $\hat{\psi}$-almost every $\chi\in \widehat{G}$ and $|\psi|$-almost every $g\in G$. For fixed $g$, the map $\chi\mapsto \overline{\chi(g)}e^{i\theta(g)}$ is continuous, therefore equation \eqref{equation phase of psi} holds for $|\psi|$-almost every $g\in G$  and every $\chi\in \mathrm{supp}(\hat{\psi})$. In particular, since we assumed that $1\in \mathrm{supp}(\hat{\psi})$, we have $e^{i\theta(g)}=1$ $|\psi|$-almost everywhere, which implies that $\psi$ is also a positive measure. It then follows $\chi(g)=1$ for $\psi\otimes\hat{\psi}$-almost every $(g,\chi)\in G\times \widehat{G}$, so by continuity $\chi(g)=1$ for all $(g,\chi)\in \mathrm{supp}(\mathrm{KD}_{\psi})=\mathrm{supp}(\psi)\times \mathrm{supp}(\hat{\psi})$. Define $H=\overline{\langle \mathrm{supp}(\psi)\rangle}$, the closed subgroup induced by the support of $\psi$. Then obviously $\mathrm{supp}(\psi)\subset H$. For any $\chi\in \mathrm{supp}(\hat{\psi})$, we have shown that $\chi(g)=1$ for all $g\in \mathrm{supp}(\psi)$. Since $\chi$ is a continuous group morphism, we must have $\chi(g)=1$ for all $g\in H$, hence $\mathrm{supp}(\hat{\psi})\subset H^{\perp}$. The result then follows from the uncertainty principle \ref{uncertainty principle}.
\end{proof}
This finishes the proof of Theorem \ref{KD positive on second countable LCA groups}.
\section{Identifying a classical fragment of quantum mechanics}
\label{Identifying classical fragment}

In this section, we address the problem of determining the quantum states (resp. Hilbert-Schmidt observables) which have a positive (resp. real) Kirkwood-Dirac distribution.
We recall the sets of interest in this setting:
\begin{itemize}
    \item $\KDpospure$, the set of KD-positive pure states.
    \item $\KDpos$, the convex set of KD-positive states.
    \item $\KDr$, the set of KD-real self-adjoint Hilbert-Schmidt operators.
\end{itemize}
We have the obvious inclusions:
\begin{equation}
    \KDpospure\subset\KDpos\subset\KDr.
\end{equation}

We begin by giving the complete description of $\KDpospure$, which follows almost immediately from Theorem \ref{KD positive on second countable LCA groups}.
Then, we use O'Connell's formula to prove Theorem \ref{existence of KD real}. It exhibits simple topological criteria that are equivalent to the fact that the classical fragment of quantum mechanics associated with the \KD distribution is non-trivial. We then give some general topological properties of the sets $\KDr$ and $\KDpos$. Finally, we prove Theorem \ref{EKD+ on connected}, which precisely describes the classical fragment of quantum mechanics for connected compact \gls{sclca} groups.

\subsection{The Kirkwood-Dirac positive pure states}

 The set $\KDpospure$ is composed of \KD positive pure states $\proj{\psi}$ with $\psi$ a norm one element of $L^2(G)$. As $L^2(G)$ embeds naturally into $\mathcal{S}'(G)$, any such state is a generalized pure state and thus the complete description of $\KDpospure$ follows from the description of $\KDpospuregen$ given in the previous section. 

\begin{thm}
    \label{KDpospure}
    A normalized function $\psi\in L^2(G)$ has a positive \KD distribution if and only if there exists a compact open subgroup $H\subset G$, and $(g,\chi,z)\in G/H\times \widehat{G}/H^{\perp}\times \mathbb{S}^1$ such that
    \begin{equation}
        \psi(g')=\psi^H_{g,\chi,z}(g')=\frac{z\chi(g')}{\sqrt{\mu_G(H)}}\mathds{1}_H(g'-g)
    \end{equation}
    for almost every $g'\in G$. Finally
    \begin{equation}
        \KDpospure=\{\proj{\psi^H_{g,\chi,1}}\},
    \end{equation}
    where $H$ runs through the compact open subgroups of $G$ and $(g,\chi)\in G/H\times \widehat{G}/H^{\perp}$ for any such $H$.
\end{thm}

\begin{proof}
Let $\proj{\psi}$ be a KD-positive pure state. Theorem \ref{KD positive on second countable LCA groups} implies that, as a tempered distribution, $\psi=\psi^H_{(g,\chi,z)}=U(g,\chi,z)\mu_H$, where $\mu_H$ is a Haar measure on some closed subgroup $H\subset G$ and for some $(g,\chi,z)\in H(G)$. Moreover, $\psi$ being a square integrable function, it can be seen as an absolutely continuous measure with respect to $\mu_G$, which implies that $\mu_G(H)>0$, and thus that $H$ is open \autocite{Bourbaki} (chapter VII). This also implies that $\mu_H=c\mathds{1}_H$ for some $c>0$ (here we identify the absolutely continuous measure with its Radon-Nikodym derivative). The fact that $\psi$ is normalized yields further $\mu_G(H)<+\infty$ and $c=\mu_G(H)^{-1/2}$. Since $H$ is open, $\mu_G\vert_H$ is a Haar measure on $H$, and $\mu_G(H)<+\infty$ therefore implies that $H$ is also compact. The Plancherel and Poisson-Tate formulas imply that $\hat{\mu}_H$ is a Haar measure on $H^{\perp}$, as well as a square integrable function. This again implies $0<\mu_{\widehat{G}}(H^{\perp})<+\infty$ and that $\hat{\mu}_H=\mu_{\widehat{G}}(H^{\perp})^{-1/2}\mathds{1}_{H^{\perp}}$. In the end, we have 
 \begin{align*}
     &\psi^H_{g,\chi,z}(g')=\frac{z\chi(g')}{\sqrt{\mu_G(H)}}\mathds{1}_H(g'-g)\\
     &\hat{\psi}^H_{g,\chi,z}(\chi')=\frac{z\chi\overline{\chi'}(g)}{\sqrt{\mu_{\widehat{G}}(H^{\perp})}}\mathds{1}_{H^{\perp}}(\chi'\overline{\chi}).
 \end{align*}
One sees that $\psi^H_{g,\chi,z}=z\psi^H_{g,\chi,1}$ and thus $\proj{\psi^H_{g,\chi,z}}=\proj{\psi^H_{g,\chi,1}}$ for any $z\in \mathbb{S}^1$. Moreover, if $(g,\chi)$ and $(\tilde{g},\tilde{\chi})$ are in the same equivalence class with respect to the subgroup $H\times H^{\perp}\subset G\times\widehat{G}$, then $\psi^H_{g,\chi,1}=\psi^H_{\tilde{g},\tilde{\chi},1}$. The set $\KDpospure$ is then exactly
\begin{equation}
\label{KD positive pure states}
    \KDpospure=\{\proj{\psi^H_{g,\chi,1}}\},
\end{equation}
where $H$ runs through the compact open subgroups of $G$ and $(g,\chi)$ runs through $G/H\times \widehat{G}/H^{\perp}$ for each $H$.
\end{proof}

The KD distributions of the elements of $\KDpospure$ are given by
\begin{equation*}
    \mathrm{KD}_{\psi^H_{g,\chi,1}}(g',\chi')=\frac{\mathds{1}_H(g'-g)\mathds{1}_{H^{\perp}}(\chi'\overline{\chi})}{\sqrt{\mu_G(H)\mu_{\widehat{G}}(H^{\perp})}}.
\end{equation*}
For a compact open subgroup $H\subset G$, the Poisson-Tate and Plancherel formulas imply that $H^{\perp}$ is also compact and open and that there exists $c>0$ such that
\begin{equation}
    \widehat{\mathds{1}_H}=c\mathds{1}_{H^{\perp}}.
\end{equation}
Integrating over $\widehat{G}$, we get by the Fourier inversion formula $c=\mu_{\widehat{G}}(H^{\perp})^{-1}$. The Plancherel formula then yields
\begin{equation}
\label{product of measures}
    \mu_G(H)\mu_{\widehat{G}}(H^{\perp})=1.
\end{equation}
It follows finally
\begin{equation}
    \label{KD distribution of clopen subgroup}
    \mathrm{KD}_{\psi^H_{g,\chi}}(g',\chi')=\mathds{1}_H(g'-g)\mathds{1}_{H^{\perp}}(\chi'\overline{\chi}).
\end{equation}

\subsection{Topological characterization of groups for which the classical fragment is non trivial}
\label{2nd result}

Note that from Theorem \ref{KDpospure}, it is possible that $\KDpospure=\emptyset$ if the group does not admit any compact open subgroup, as for $\mathbb{R}^n$ for example. In that case there are no KD-positive pure states, but a priori there could still exist KD-positive states that are non-pure and non-zero KD-real Hilbert-Schmidt observables. However, Theorem \ref{existence of KD real} forbids such a situation. 
To prove it, we firstly establish the following lemma which describes the set $\KDr$ using the support of the characteristic function. It is a direct consequence of O'Conell's formula \eqref{KD and characteristic}.

\begin{lemma}
\label{support of characteristic function}
    Let $G$ be a \gls{sclca} group. Then $\KDr$ is the set of self-adjoint Hilbert-Schmidt operators $A$ for which the characteristic function has the following property:
    \begin{equation}
    \label{support of characteristic function equation}
        \mathrm{supp}(\mathfrak{X}^0_{A})\subset\{(g,\chi)\in G\times\widehat{G}:\chi(g)=1\},
    \end{equation}
    where $\mathrm{supp}(\mathfrak{X}^0_{A})$ is the essential support of the  function $\mathfrak{X}^0_{A}\in L^2(G\times \widehat{G})$.
\end{lemma}

\begin{proof}
     Let $A$ be a self-adjoint Hilbert-Schmidt operator.
    O'Conell's formula is
    \begin{equation}
        \mathrm{KD}_{A}=\mathcal{F}_{\mathrm{symp}}(\mathfrak{X}^1_{A})
    \end{equation}
    We also have from \eqref{KD conjugate} and since $A$ is self-adjoint
    \begin{equation}
        \overline{\mathrm{KD}_A}=\mathcal{F}_{\mathrm{symp}}(\mathfrak{X}_A^0).
    \end{equation}
    Thus $\mathrm{KD}_A$ is a.e. real if and only if
    \begin{equation}
        \mathcal{F}_{\mathrm{symp}}(\mathfrak{X}^1_{A})=\mathcal{F}_{\mathrm{symp}}(\mathfrak{X}_A^0).
    \end{equation}
    Taking the inverse Fourier transform, this is equivalent to 
\begin{equation*}
\overline{\chi(g)}\mathfrak{X}^0_{A}(g,\chi)=\mathfrak{X}^0_{A}(g,\chi)
\end{equation*}
almost-everywhere.
Since $\{(g,\chi)\in G\times \widehat{G}: \chi(g)=1\}$ is a closed subset of $G\times \widehat{G}$, this is equivalent to $\mathrm{supp}(\mathfrak{X}^0_{A})\subset \{(g,\chi)\in G\times \widehat{G}: \chi(g)= 1\}$, and the result follows.
\end{proof}

We recall that the identity component $G_0$ of a \gls{sclca} group $G$ is the connected component of $G$ that contains $0$. It is a standard fact that it is a closed subgroup of $G$. However, it is not necessarily open. For clarity we state Theorem \ref{existence of KD real} once more:
\begingroup
\renewcommand{\thethm}{1.2}
\begin{thm}
    Let $G$ be a \gls{sclca} group and let $G_0$ be its identity component. The following are equivalent:
    \begin{enumerate}
        \item $G_0$ is compact.
        \item $G$ admits a compact open subgroup.
        \item $\KDpospure\neq \emptyset$.
        \item $\KDpos\neq \emptyset$.
        \item $\KDr\neq \{0\}$.
    \end{enumerate}
\end{thm}
\endgroup
\begin{proof}
    The fact that 1 implies 2 follows from the structure theorem \eqref{structure theorem}: $G\simeq \mathbb{R }^n\times H_1$, where $H_1$ admits a compact open set. The fact that $G_0$ is compact implies $n=0$ so $G\simeq H_1$ admits indeed a compact open subgroup.
    
    The fact that 2 implies 3 follows from Theorem \ref{KDpospure}.
    The implications 3 implies 4 implies 5 are trivial.
    
    We finish by showing that 5 implies 1.
     Let $A\in\KDr\setminus\{0\}$. Then $\mathfrak{X}^0_{A}$ is a non-zero element of $L^2(G\times \widehat{G})$, and it satisfies \eqref{support of characteristic function equation}. In particular, we must have
    \begin{equation}
    \label{positive product measure}
        0<\mu_G\otimes\mu_{\widehat{G}}(\{(g,\chi)\in G\times \widehat{G}:\chi(g)=1\})=\int_{\widehat{G}}\mu_G(\langle \chi\rangle ^{\perp})d\mu_{\widehat{G}}(\chi)
    \end{equation}
    by Tonelli's theorem, where for any $\chi\in \widehat{G}$, $\langle \chi\rangle ^{\perp}=\{g\in G:\chi(g)= 1\}$. This equation implies 
    \begin{equation}
    \label{positive hat measure}
        \mu_{\widehat{G}}(\{\chi\in \widehat{G}:\mu_G(\langle \chi\rangle ^{\perp})>0\})>0.
    \end{equation}
    Let $\chi\in \widehat{G}$ such that $\mu_G(\langle \chi\rangle ^{\perp})>0$. Since $\langle \chi\rangle ^{\perp}$ is a closed subgroup of $G$, this implies that it is also open, and thus $G_0\subset \langle \chi\rangle ^{\perp}$. In particular $\chi\vert_{G_0}= 1$. Thus
    \begin{equation}
         \{\chi\in \widehat{G}:\mu_G(\langle \chi\rangle ^{\perp})>0\}\subset G_0^{\perp},
    \end{equation}
    so \eqref{positive hat measure} implies $\mu_{\widehat{G}}(G_0^{\perp})>0$, hence $G_0^{\perp}$ is open in $\widehat{G}$, so $\widehat{G}/G_0^{\perp}$ is discrete. We finally obtain that $G_0\simeq (G_0^{\perp})^{\perp}\simeq \widehat{\widehat{G}/G_0^{\perp}}$ is  a compact \gls{sclca} group.
\end{proof}

When the \gls{sclca} group $G$ is compact and connected, a precise description of $\KDr$ and $\KDpos$ can be given. This is the content of Theorem \ref{EKD+ on connected} that we prove in section \ref{3rd result}. First, we need some topological properties of both of those sets.

\subsection{Topological properties of the sets of KD-real Hilbert-Schmidt observables and of KD-positive states}

It is immediate to see that $\KDr$ is a real vector space. Moreover, the subset of $L^2(G\times \widehat{G})$ consisting of real-valued functions is closed. The set of self-adjoint Hilbert-Schmidt operators is closed in the Hilbert-Schmidt topology. As the \KD distribution is continuous from $\mathcal{T}_2(G)$ to $L^2(G\times \widehat{G})$ for the natural topologies, it follows that $\KDr$ is closed in the Hilbert-Schmidt topology.

It is also immediate to see that the set $\KDpos$ is convex. Let's us give some of its general properties. We recall that the weak topology on $\mathcal{T}_1(G)$ is defined by the weakest topology so that, for any bounded operator $A$, the map $\rho\mapsto \tr(\rho A)$ is continuous. Since those maps are continuous for the trace norm topology, the weak topology is weaker than the trace norm topology. Being inspired by the results of \autocite{Lami_Shirokov}, we prove the following:
\begin{prop}
\label{KD pos closed}
    Let $G$ be a \gls{sclca} group. The set $\KDpos$ is closed in the weak topology.
\end{prop}

\begin{proof}
    For $\rho$ a state on $L^2(G)$, its characteristic function $\mathfrak{X}^0_{\rho}$ belongs to $L^2(G\times \widehat{G})$. By formula \eqref{KD and characteristic}, and by Bochner's theorem (see \autocite{Rudin}), $\mathrm{KD}_{\rho}$ is positive if and only if $\mathfrak{X}^1_{\rho}$ is positive semi-definite, i.e. $\KDpos$ is equal to
    {\footnotesize
    \begin{equation*}
        \bigcap_{N\in \mathbb{N}}\bigcap_{z\in \mathbb{C}^N}\bigcap_{(g_i,\chi_i)\in (G\times \widehat{G})^N}\{\rho: \text{ state},\: \sum_{i,j=1}^Nz_i\overline{z_j}\overline{\chi_i(g_i)}\overline{\chi_j(g_j)}\chi_i(g_j)\chi_j(g_i)\tr(U(g_i-g_j,\chi_i\overline{\chi_j},1)\rho)\geq 0\}.
    \end{equation*}
    }
    Since the operators $U(g,\chi,1)$ are bounded for any $(g,\chi)\in G\times \widehat{G}$, and since the set of states is closed for the weak topology, $\KDpos$ is indeed closed for the weak topology.
\end{proof}
This in particular implies that $\KDpos$ is closed in the trace norm topology. Moreover, it is also a bounded set since for any $\rho\in \KDpos$, $||\rho||_1=\tr(|\rho|)=\tr(\rho)=1$. For a \gls{sclca} group $G$, $\mathcal{T}_1(G)$ is a separable Banach space, and can be seen as the dual of the Banach space of compact operators on $L^2(G)$, endowed with the operator norm. $\mathcal{T}_1(G)$ is therefore a separable dual Banach space. For such spaces, a generalization of Krein-Millman theorem holds, known as Lindenstrauss's theorem, see \autocite{Phelps}. Namely, any closed, bounded and convex subset of a separable dual Banach space is equal to the closed convex hull of its extreme points. In our setting, this yields:

\begin{corollary}
\label{KD pos closed convex}
    Let $G$ be a \gls{sclca} group. The set $\KDpos$ is closed in the trace norm topology, and is equal to the closed convex hull of its extreme points.
\end{corollary}

Let $\mathrm{Ext}(\KDpos)$ be the set of extreme points of $\KDpos$. Edgar's theorem \autocite{Edgar}, a generalization of Choquet's theorem to this non compact setting, states that for any $\rho\in \KDpos$, there exists a Borel probability measure $\mu_{\rho}$ on $\KDpos$, supported by $\mathrm{Ext}(\KDpos)$, such that
\begin{equation}
    \tr(\rho A)=\int_{\mathrm{Ext(\KDpos)}}\tr(\rho_0A)d\mu_{\rho}(\rho_0)
\end{equation}
for any bounded operator $A$ on $L^2(G)$. Those results allow us to reduce the descprition of $\KDpos$ to the descrpition of its extreme points. As the pure states are extreme points of the set of states, we have immediately
\begin{equation}
    \KDpospure\subset\mathrm{Ext}(\KDpos).
\end{equation}
One can then ask: does
\begin{equation}
\label{equality between pure and extreme}
    \KDpospure=\mathrm{Ext}(\KDpos)
\end{equation}
hold ?

In \autocite{KD_finite_groups}, examples are provided of finite groups for which this equality holds (cyclic groups of prime power order), and for which it does not (e.g. $\mathbb{Z}/2\mathbb{Z}\times \mathbb{Z}/2\mathbb{Z}$).
In the next subsection, we prove that it holds under the assumption that $G$ is a connected and compact \gls{sclca} group.

\subsection{Complete description of the classical fragment for connected compact SCLCA groups}
\label{3rd result}

In this subsection, we prove Theorem \ref{EKD+ on connected}. Combined with Theorem \ref{existence of KD real}, this solves the problem of the description of the classical fragment for any connected \gls{sclca} group. Of course by Pontryagin duality, Theorem \ref{EKD+ on connected} can be applied for groups for which $\widehat{G}$ is compact and connected, i.e. discrete torsion free groups.

As before we state again the result we prove:

\begingroup
\renewcommand{\thethm}{1.3}
\begin{thm}
    Let $G$ be a connected, compact, \gls{sclca} group. Then
    \begin{equation}
        \KDr=\overline{\mathrm{span}_{\mathbb{R}}(\KDpospure)} \text{ and }
        \KDpos=\mathrm{conv}(\KDpospure),
    \end{equation}
    where the former closure is taken with respect to the Hilbert-Schmidt topology while the latter closed convex hull is taken in the trace norm topology.
\end{thm}
\endgroup
The proof of this theorem uses once more Lemma \ref{support of characteristic function}.
\begin{proof}
    We prove that $\KDpos=\mathrm{conv}(\KDpospure)$. The proof for $\KDr$ is similar. Let $\rho\in \KDpos$. Then as in the proof of Theorem \ref{existence of KD real}, the characteristic function of $\rho$ is non-zero and is supported by
    \begin{equation*}
        \{(g,\chi)\in G\times \widehat{G}:\chi(g)=1\}=\bigcup_{\chi\in \widehat{G}}\langle \chi\rangle ^{\perp}\times \{\chi\}.
    \end{equation*}
     Recall that for any $\chi\in \widehat{G}$, $\langle \chi\rangle ^{\perp}$ is a closed subgroup, and has positive Haar measure if and only if it is also open. By connectedness, this implies that $\langle \chi\rangle ^{\perp}$ has positive Haar measure if and only if $\langle \chi\rangle ^{\perp}=G$, or equivalently $\chi=1$. Moreover, since $G$ is compact and second countable, $\widehat{G}$ is discrete and second countable, and therefore countable. It follows that
     \begin{equation*}
       \mu_G\otimes\mu_{\widehat{G}}\left(\bigcup_{\chi\neq 1}\langle \chi\rangle ^{\perp}\times \{\chi\}\right)=0,
    \end{equation*}
    and thus the essential support of $\mathfrak{X}^0_{\rho}$ is in fact contained in $G\times \{1\}$, i.e. there exists $u_{\rho}\in L^2(G)$ such that
    \begin{equation*}
        \mathfrak{X}^0_{\rho}(g,\chi)=u_{\rho}(g)\delta_{\chi,1}.
    \end{equation*}
    Plugging this expression into formula \eqref{KD and characteristic} yields
    \begin{equation}
        \mathrm{KD}_{\rho}(g,\chi)=\int_Gu_{\rho}(g')\chi(g')d\mu_G(g'),
    \end{equation}
    so $\mathrm{KD}_{\rho}$ does not depend on the variable $g\in G$. Now by assumption it is also a positive square integrable function on $G\times \widehat{G}$, so since $\widehat{G}$ is countable there exists $(\rho_{\chi})_{\chi\in \widehat{G}}\subset \mathbb{R}^+$ such that
    \begin{equation}
        \mathrm{KD}_{\rho}=\sum_{\chi\in \widehat{G}}\rho_{\chi}\mathds{1}_{\chi},
    \end{equation}
    and 
    \begin{equation}
        \sum_{\chi\in \widehat{G}}\rho_{\chi}^2<+\infty.
    \end{equation}
    Taking the inverse of the \KD distribution yields
    \begin{equation}
        \rho=\sum_{\chi\in \widehat{G}}\rho_{\chi}\proj{\chi}.
    \end{equation}
    The condition $\tr(\rho)=1$ reads $\sum_{\chi\in \widehat{G}}\rho_{\chi}=1$, hence $\rho\in \mathrm{conv}(\KDpospure)$. Finally, we have
    \begin{equation}
        \KDpospure\subset \KDpos\subset\mathrm{conv}(\KDpospure),
    \end{equation}
    and since $\KDpos$ is a closed convex set by Corollary \ref{KD pos closed convex}, the result follows.
\end{proof}

In the compact connected case, the KD-real observables are therefore exactly the observables that are diagonal in the Fourier basis.

\begin{rmk}
    The set $\KDpos$ is in general not closed in the Hilbert-Schmidt topology. Indeed take $G=\mathbb{S}^1$. Then the characters (or Fourier basis) are given by $(\ket{k})_{k\in \mathbb{Z}}$, where $\braket{z|k}=z^k$. By Theorem \ref{EKD+ on connected} each state of the form
    \begin{equation*}
        \rho_a=(1-e^{-a})\sum_{k=0}^{+\infty}e^{-ak}\proj{k},\: a>0
    \end{equation*}
    belongs to $\KDpos$. Moreover a simple computation shows that their Hilbert-Schmidt norm is
    \begin{equation*}
        ||\rho_a||_2^2=\frac{1-e^{-a}}{1+e^{-a}},
    \end{equation*}
    so it tends to $0$ as $a$ goes to $0$. Since $0\notin \KDpos$, the set isn't closed for the Hilbert-Schmidt topology.
\end{rmk}

\subsection{Positivity of the Kirkwood-Dirac and Wigner distributions of quantum states: discussion}
\label{KD vs Wigner}

We finish this article by comparing positivity properties of quantum states with respect to the Wigner and \KD distributions, depending on the structure of the groups.

As already mentioned, the \KD distribution can be defined on any \gls{sclca} group, as opposed to the Wigner distribution. For this reason, we will assume that the \gls{sclca} group $G$ has an invertible doubling map $g\mapsto g+g$, so that the Wigner distribution is well-defined \autocite{Crann}. Let us denote by $\Wpos$ the set of Wigner-positive sates and by $\Wpospure$ the set of Wigner-positive pure states. We compare those sets with the sets $\KDpos$ and $\KDpospure$.

The structure theorem implies that there exist a non-negative integer $n$ and a \gls{sclca} group $H$ admitting a compact open subgroup such that
\begin{equation*}
    G\simeq \mathbb{R}^n\times H.
\end{equation*}
For the Wigner distribution, the case where $H$ is trivial ($G\simeq \mathbb{R}^n$) has been extensively studied. The set $\Wpospure$ has been identified in the celebrated Hudson theorem \autocite{Hudson}. Yet, simply characterizing $\Wpos$ is still an open problem \autocite{ Bröcker_Werner, Mandilara_Karpov_Cerf, Cerf_Van_Herstraeten_Beam_splitter}. 

From Theorem \ref{existence of KD real}, if $n\geq 1$ in the above decomposition, then $\KDpos=\emptyset$, hence we will assume from now on that $n=0$ (equivalently $G$ admits a compact open subgroup). This case is precisely the one where the Wigner distributions of \textit{pure} states have been recently partially studied in \autocite{Crann}, and then in full generality in \autocite{NicolaRiccardi2025}. To the best of our knowledge, the Wigner distributions of mixed states have not been studied in this generality. For pure states, there are two subcases.

The first is when the invertible doubling map $g\mapsto g+g$ is not measure preserving (with respect to the Haar measure). In this case, Theorem 1.3 of \autocite{NicolaRiccardi2025} asserts $\Wpospure=\emptyset$. However $\KDpospure\neq \emptyset$, as ensured by Theorem \ref{existence of KD real}.

The second is when the doubling map is measure preserving, for example for any finite abelian group with odd order. In this case, Theorem \ref{KDpospure} of the present paper and Theorem 1.1 of \autocite{NicolaRiccardi2025} imply that we have the strict inclusion
\begin{equation}
    \KDpospure\subsetneq\Wpospure.
\end{equation}
 A natural question to ask is whether a similar inclusion holds for mixed states. Despite having few available results on the Wigner distributions of mixed states, we can give a positive answer for groups for which 
\begin{equation}
\label{convex equality 2}
    \KDpos=\mathrm{conv}(\KDpospure)
\end{equation}
holds, in addition to the above assumptions. Indeed we get in this case
\begin{equation}
    \KDpos=\mathrm{conv}(\KDpospure)\subset \mathrm{conv}(\Wpospure)\subset \Wpos,
\end{equation}
since $\Wpos$ is a convex set.
In particular, we can give three families of groups for which \KD positivity implies Wigner positivity. The first is composed of the groups $\mathbb{Z}/p^k\mathbb{Z}$, for $p\geq 3$ a prime number and $k\geq 1$ \autocite{KD_finite_groups}. The second is the family of compact connected \gls{sclca} groups with an invertible and measure preserving doubling map, as a consequence of Theorem \ref{EKD+ on connected}. The third consists of groups which are dual to those in the second family: discrete torsion-free \gls{sclca} groups with an invertible and measure preserving doubling map. An example of a group in the third family is the additive group of rational numbers $\mathbb{Q}$ endowed with the discrete topology. This means that its dual, $\mathbb{A}_{\mathbb{Q}}/\mathbb{Q}$ \autocite{vourdasHarmonicAnalysisRational2012} is in the second family.

One could finally ask what happens when considering generalized states. The simplest example to study is $G=\mathbb{R}$. In that case the KD-positive generalized pure states are the position basis, the momentum basis and the GKP states (we refer to the Introduction for more details). It is easy to compute that the position and momentum bases have positive Wigner distributions. However, it is well-known that the GKP states are ``highly Wigner non-positive'' \autocite{GKP}.

In the end there isn't any clear and direct relation between KD positivity and Wigner positivity in general. However these two notions are not unrelated either.

\emph{Acknowledgments:} This work was supported in part by the CNRS through the MITI interdisciplinary programs. The author acknowledges the support of the CDP C2EMPI, as well as the French State under the France-2030 programme, the University of Lille, the Initiative of Excellence of the University of Lille, the European Metropolis of Lille for their funding and support of the R-CDP-24-004-C2EMPI project.
The author thanks S. De Bièvre, C. Langrenez and D. Radchenko for introducing the \KD distribution to him and for enlightening discussions.

This work is part of the author's PhD thesis.
\printbibliography

@book{Bourbaki,
  title = {Intégration},
  author = {Bourbaki, Nicolas},
  date = {2007},
  series = {Eléments de mathématique},
  edition = {Reproduction en fac-similé},
  number = {6},
  publisher = {Springer},
  location = {Berlin},
  isbn = {978-3-540-35324-9},
  langid = {french}
}

@article{Bröcker_Werner,
  title = {Mixed States with Positive {{Wigner}} Functions},
  author = {Bröcker, T. and Werner, R. F.},
  date = {1995-01-01},
  journaltitle = {Journal of Mathematical Physics},
  shortjournal = {Journal of Mathematical Physics},
  volume = {36},
  number = {1},
  pages = {62--75},
  issn = {0022-2488},
  doi = {10.1063/1.531326},
  url = {https://doi.org/10.1063/1.531326},
  urldate = {2025-05-21}
}

@article{Bruhat_distributions,
  title = {Distributions sur un groupe localement compact et applications à l'étude des représentations des groupes p-adiques},
  author = {Bruhat, François},
  date = {1961},
  journaltitle = {Bulletin de la Société mathématique de France},
  shortjournal = {Bul. Soc. Math. France},
  volume = {79},
  pages = {43--75},
  issn = {0037-9484, 2102-622X},
  doi = {10.24033/bsmf.1559},
  url = {http://www.numdam.org/item?id=BSMF_1961__89__43_0},
  urldate = {2025-03-04},
  langid = {french}
}

@article{Cahill_Glauber,
  title = {Density {{Operators}} and {{Quasiprobability Distributions}}},
  author = {Cahill, K. E. and Glauber, R. J.},
  date = {1969-01-25},
  journaltitle = {Physical Review},
  shortjournal = {Phys. Rev.},
  volume = {177},
  number = {5},
  pages = {1882--1902},
  issn = {0031-899X},
  doi = {10.1103/PhysRev.177.1882},
  url = {https://link.aps.org/doi/10.1103/PhysRev.177.1882},
  urldate = {2025-03-04},
  langid = {english}
}

@article{Cerf_Van_Herstraeten_Beam_splitter,
  title = {Quantum {{Wigner}} Entropy},
  author = {Van Herstraeten, Zacharie and Cerf, Nicolas J.},
  date = {2021-10-13},
  journaltitle = {Physical Review A},
  shortjournal = {Phys. Rev. A},
  volume = {104},
  number = {4},
  pages = {042211},
  publisher = {American Physical Society},
  doi = {10.1103/PhysRevA.104.042211},
  url = {https://link.aps.org/doi/10.1103/PhysRevA.104.042211},
  urldate = {2025-09-05}
}

@incollection{Cohn,
  title = {Convergence},
  booktitle = {Measure {{Theory}}: {{Second Edition}}},
  author = {Cohn, Donald L.},
  editor = {Cohn, Donald L.},
  date = {2013},
  pages = {79--111},
  publisher = {Springer},
  location = {New York, NY},
  doi = {10.1007/978-1-4614-6956-8_3},
  url = {https://doi.org/10.1007/978-1-4614-6956-8_3},
  urldate = {2025-04-17},
  isbn = {978-1-4614-6956-8},
  langid = {english}
}

@article{Crann,
  title = {{{GAUSSIAN QUANTUM INFORMATION OVER GENERAL QUANTUM KINEMATICAL SYSTEMS I}}: {{GAUSSIAN STATES}}},
  shorttitle = {{{GAUSSIAN QUANTUM INFORMATION OVER GENERAL QUANTUM KINEMATICAL SYSTEMS I}}},
  author = {Bény, Cédric and Crann, Jason and Lee, Hun Hee and Park, Sang-Jun and Youn, Sang-Gyun},
  date = {2025-03-20},
  journaltitle = {Letters in Mathematical Physics},
  shortjournal = {Lett Math Phys},
  volume = {115},
  number = {2},
  pages = {32},
  issn = {1573-0530},
  doi = {10.1007/s11005-025-01908-1},
  url = {https://doi.org/10.1007/s11005-025-01908-1},
  urldate = {2025-06-16},
  langid = {english}
}

@article{Dirac,
  title = {On the {{Analogy Between Classical}} and {{Quantum Mechanics}}},
  author = {Dirac, P. A. M.},
  date = {1945-04-01},
  journaltitle = {Reviews of Modern Physics},
  shortjournal = {Rev. Mod. Phys.},
  volume = {17},
  number = {2--3},
  pages = {195--199},
  issn = {0034-6861},
  doi = {10.1103/RevModPhys.17.195},
  url = {https://link.aps.org/doi/10.1103/RevModPhys.17.195},
  urldate = {2025-04-01},
  langid = {english}
}

@article{Donoho_Stark,
  title = {Uncertainty {{Principles}} and {{Signal Recovery}}},
  author = {Donoho, David L. and Stark, Philip B.},
  date = {1989-06},
  journaltitle = {SIAM Journal on Applied Mathematics},
  shortjournal = {SIAM J. Appl. Math.},
  volume = {49},
  number = {3},
  pages = {906--931},
  publisher = {{Society for Industrial and Applied Mathematics}},
  issn = {0036-1399},
  doi = {10.1137/0149053},
  url = {https://epubs.siam.org/doi/10.1137/0149053},
  urldate = {2025-04-17}
}

@article{Edgar,
  title = {A Non Compact {{Choquet}} Theorem},
  author = {Edgar, G A},
  date = {1975-06},
  journaltitle = {Proceedings of the American Mathematical Society},
  volume = {49},
  number = {2},
  pages = {354--358},
  langid = {english}
}

@incollection{FeichtingerKozek98,
  title = {Quantization of {{TF}} Lattice-Invariant Operators on Elementary {{LCA}} Groups},
  booktitle = {Gabor {{Analysis}} and {{Algorithms}}: {{Theory}} and {{Applications}}},
  author = {Feichtinger, Hans G. and Kozek, Werner},
  editor = {Feichtinger, Hans G. and Strohmer, Thomas},
  date = {1998},
  pages = {233--266},
  publisher = {Birkhäuser},
  location = {Boston, MA},
  doi = {10.1007/978-1-4612-2016-9_8},
  url = {https://doi.org/10.1007/978-1-4612-2016-9_8},
  isbn = {978-1-4612-2016-9},
  langid = {english}
}

@article{Ferrie,
  title = {Quasi-Probability Representations of Quantum Theory with Applications to Quantum Information Science},
  author = {Ferrie, Christopher},
  date = {2011-11-01},
  journaltitle = {Reports on Progress in Physics},
  shortjournal = {Rep. Prog. Phys.},
  volume = {74},
  number = {11},
  eprint = {1010.2701},
  eprinttype = {arXiv},
  eprintclass = {quant-ph},
  pages = {116001},
  issn = {0034-4885, 1361-6633},
  doi = {10.1088/0034-4885/74/11/116001}
}

@book{Folland,
  title = {Harmonic {{Analysis}} in {{Phase Space}}. ({{AM-122}})},
  author = {Folland, Gerald B.},
  date = {1989},
  eprint = {j.ctt1b9rzs2},
  eprinttype = {jstor},
  publisher = {Princeton University Press},
  url = {https://www.jstor.org/stable/j.ctt1b9rzs2},
  urldate = {2025-07-30},
  isbn = {978-0-691-08528-9}
}

@article{GKP,
  title = {Encoding a Qubit in an Oscillator},
  author = {Gottesman, Daniel and Kitaev, Alexei and Preskill, John},
  date = {2001-06-11},
  journaltitle = {Physical Review A},
  shortjournal = {Phys. Rev. A},
  volume = {64},
  number = {1},
  pages = {012310},
  publisher = {American Physical Society},
  doi = {10.1103/PhysRevA.64.012310},
  url = {https://link.aps.org/doi/10.1103/PhysRevA.64.012310},
  urldate = {2025-05-15}
}

@article{Gross_Wigner,
  title = {Hudson's {{Theorem}} for Finite-Dimensional Quantum Systems},
  author = {Gross, D.},
  date = {2006-12-01},
  journaltitle = {Journal of Mathematical Physics},
  volume = {47},
  number = {12},
  eprint = {quant-ph/0602001},
  eprinttype = {arXiv},
  pages = {122107},
  issn = {0022-2488, 1089-7658},
  doi = {10.1063/1.2393152},
  langid = {english}
}

@article{Hennings,
  title = {The {{Weyl}} Transformation and Quantisation for Locally Compact {{Abelian}} Groups},
  author = {Hennings, M A},
  date = {1985-12-01},
  journaltitle = {Publ. Res. Inst. Math. Sci.},
  volume = {21},
  number = {6},
  pages = {1223--1235},
  issn = {0034-5318},
  doi = {10.2977/prims/1195178513},
  url = {https://doi.org/10.2977/prims/1195178513},
  urldate = {2025-07-31}
}

@article{Hudson,
  title = {When Is the Wigner Quasi-Probability Density Non-Negative?},
  author = {Hudson, R. L.},
  date = {1974-10-01},
  journaltitle = {Reports on Mathematical Physics},
  shortjournal = {Reports on Mathematical Physics},
  volume = {6},
  number = {2},
  pages = {249--252},
  issn = {0034-4877},
  doi = {10.1016/0034-4877(74)90007-X},
  url = {https://www.sciencedirect.com/science/article/pii/003448777490007X},
  urldate = {2025-05-02}
}

@article{Husimi,
  title = {Some {{Formal Properties}} of the {{Density Matrix}}},
  author = {Husimi, Kôdi},
  date = {1940},
  journaltitle = {Proceedings of the Physico-Mathematical Society of Japan. 3rd Series},
  volume = {22},
  number = {4},
  pages = {264--314},
  doi = {10.11429/ppmsj1919.22.4_264}
}

@article{KD_finite_groups,
  title = {The {{Kirkwood-Dirac Representation Associated}} to the {{Fourier Transform}} for {{Finite Abelian Groups}}: {{Positivity}}},
  shorttitle = {The {{Kirkwood-Dirac Representation Associated}} to the {{Fourier Transform}} for {{Finite Abelian Groups}}},
  author = {De Bièvre, Stephan and Langrenez, Christopher and Radchenko, Danylo},
  date = {2025-09-02},
  journaltitle = {Annales Henri Poincaré},
  shortjournal = {Ann. Henri Poincaré},
  issn = {1424-0661},
  doi = {10.1007/s00023-025-01614-7},
  url = {https://doi.org/10.1007/s00023-025-01614-7},
  urldate = {2025-09-08},
  langid = {english}
}

@article{KD_review,
  title = {Properties and {{Applications}} of the {{Kirkwood-Dirac Distribution}}},
  author = {Arvidsson-Shukur, David R. M. and Jr, William F. Braasch and Bievre, Stephan De and Dressel, Justin and Jordan, Andrew N. and Langrenez, Christopher and Lostaglio, Matteo and Lundeen, Jeff S. and Halpern, Nicole Yunger},
  date = {2024-12-01},
  journaltitle = {New Journal of Physics},
  shortjournal = {New J. Phys.},
  volume = {26},
  number = {12},
  eprint = {2403.18899},
  eprinttype = {arXiv},
  eprintclass = {quant-ph},
  pages = {121201},
  issn = {1367-2630},
  doi = {10.1088/1367-2630/ada05d}
}

@online{KD_simulability,
  title = {Kirkwood-{{Dirac Nonpositivity}} Is a {{Necessary Resource}} for {{Quantum Computing}}},
  author = {Thio, Jonathan J. and Yang, Songqinghao and Bièvre, Stephan De and Barnes, Crispin H. W. and Arvidsson-Shukur, David R. M.},
  date = {2025-06-09},
  eprint = {2506.08092},
  eprinttype = {arXiv},
  eprintclass = {quant-ph},
  doi = {10.48550/arXiv.2506.08092},
  pubstate = {prepublished}
}

@article{Kirkwood,
  title = {Quantum {{Statistics}} of {{Almost Classical Assemblies}}},
  author = {Kirkwood, John G.},
  date = {1933-07-01},
  journaltitle = {Physical Review},
  shortjournal = {Phys. Rev.},
  volume = {44},
  number = {1},
  pages = {31--37},
  issn = {0031-899X},
  doi = {10.1103/PhysRev.44.31},
  url = {https://link.aps.org/doi/10.1103/PhysRev.44.31},
  urldate = {2025-04-01},
  langid = {english}
}

@article{Kohn_Nirenberg_original,
  title = {An Algebra of Pseudo-Differential Operators},
  author = {Kohn, J. J. and Nirenberg, L.},
  date = {1965},
  journaltitle = {Communications on Pure and Applied Mathematics},
  volume = {18},
  number = {1--2},
  pages = {269--305},
  issn = {1097-0312},
  doi = {10.1002/cpa.3160180121},
  url = {https://onlinelibrary.wiley.com/doi/abs/10.1002/cpa.3160180121},
  langid = {english}
}

@online{Kohn_Nirenberg_quantization,
  title = {Quantization on {{Groups}} and {{Garding}} Inequality},
  author = {Benedetto, Lino and Kammerer, Clotilde Fermanian and Fischer, Véronique},
  date = {2024-02-12},
  eprint = {2307.15352},
  eprinttype = {arXiv},
  eprintclass = {math},
  doi = {10.48550/arXiv.2307.15352},
  pubstate = {prepublished}
}

@article{Lami_Shirokov,
  title = {Attainability and Lower Semi-Continuity of the Relative Entropy of Entanglement, and Variations on the Theme},
  author = {Lami, Ludovico and Shirokov, Maksim E.},
  date = {2023-12},
  journaltitle = {Annales Henri Poincaré},
  shortjournal = {Ann. Henri Poincaré},
  volume = {24},
  number = {12},
  eprint = {2105.08091},
  eprinttype = {arXiv},
  eprintclass = {quant-ph},
  pages = {4069--4137},
  issn = {1424-0637, 1424-0661},
  doi = {10.1007/s00023-023-01313-1},
  langid = {english}
}

@article{Mandilara_Karpov_Cerf,
  title = {Extending {{Hudson}}'s Theorem to Mixed Quantum States},
  author = {Mandilara, A. and Karpov, E. and Cerf, N. J.},
  date = {2009-06-03},
  journaltitle = {Physical Review A},
  shortjournal = {Phys. Rev. A},
  volume = {79},
  number = {6},
  pages = {062302},
  publisher = {American Physical Society},
  doi = {10.1103/PhysRevA.79.062302},
  url = {https://link.aps.org/doi/10.1103/PhysRevA.79.062302},
  urldate = {2025-07-26}
}

@article{Margeneau_Hill,
  title = {Correlation between {{Measurements}} in {{Quantum Theory}}:},
  shorttitle = {Correlation between {{Measurements}} in {{Quantum Theory}}},
  author = {Margenau, Henry and Hill, Robert Nyden},
  date = {1961-11-01},
  journaltitle = {Progress of Theoretical Physics},
  shortjournal = {Progress of Theoretical Physics},
  volume = {26},
  number = {5},
  pages = {722--738},
  issn = {0033-068X},
  doi = {10.1143/PTP.26.722},
  url = {https://doi.org/10.1143/PTP.26.722},
  urldate = {2025-05-21}
}

@online{NicolaRiccardi2025,
  title = {The {{Hudson}} Theorem in {{LCA}} Groups and Infinite Quantum Spin Systems},
  author = {Nicola, Fabio and Riccardi, Federico},
  date = {2025-07-17},
  eprint = {2507.13154},
  eprinttype = {arXiv},
  eprintclass = {math-ph},
  doi = {10.48550/arXiv.2507.13154},
  pubstate = {prepublished},
  version = {1}
}

@incollection{OConell,
  title = {Quantum {{Distribution Functions}} in {{Non-Equilibrium Statistical Mechanics}}},
  booktitle = {Frontiers of {{Nonequilibrium Statistical Physics}}},
  author = {O’Connell, R. F.},
  editor = {Moore, Gerald T. and Scully, Marlan O.},
  date = {1986},
  pages = {83--95},
  publisher = {Springer US},
  location = {Boston, MA},
  doi = {10.1007/978-1-4613-2181-1_5},
  url = {https://doi.org/10.1007/978-1-4613-2181-1_5},
  urldate = {2025-03-26},
  isbn = {978-1-4613-2181-1},
  langid = {english}
}

@article{OConell_Wang,
  title = {A New Parametrized Quantum Distribution Function and Its Time Development},
  author = {O'Connell, R.F. and Wang, Lipo},
  date = {1985-01},
  journaltitle = {Physics Letters A},
  shortjournal = {Physics Letters A},
  volume = {107},
  number = {1},
  pages = {9--12},
  issn = {03759601},
  doi = {10.1016/0375-9601(85)90235-X},
  url = {https://linkinghub.elsevier.com/retrieve/pii/037596018590235X},
  urldate = {2025-03-04},
  langid = {english}
}

@article{Osborne_Schwartz_functions,
  title = {On the {{Schwartz-Bruhat}} Space and the {{Paley-Wiener}} Theorem for Locally Compact Abelian Groups},
  author = {Osborne, M.Scott},
  date = {1975-05},
  journaltitle = {Journal of Functional Analysis},
  shortjournal = {Journal of Functional Analysis},
  volume = {19},
  number = {1},
  pages = {40--49},
  issn = {00221236},
  doi = {10.1016/0022-1236(75)90005-1},
  url = {https://linkinghub.elsevier.com/retrieve/pii/0022123675900051},
  urldate = {2025-03-04},
  langid = {english}
}

@article{Pashayan_Wallman_Bartlett,
  title = {Estimating {{Outcome Probabilities}} of {{Quantum Circuits Using Quasiprobabilities}}},
  author = {Pashayan, Hakop and Wallman, Joel J. and Bartlett, Stephen D.},
  date = {2015-08-10},
  journaltitle = {Physical Review Letters},
  shortjournal = {Phys. Rev. Lett.},
  volume = {115},
  number = {7},
  pages = {070501},
  publisher = {American Physical Society},
  doi = {10.1103/PhysRevLett.115.070501},
  url = {https://link.aps.org/doi/10.1103/PhysRevLett.115.070501},
  urldate = {2025-07-26}
}

@article{Phelps,
  title = {Dentability and Extreme Points in {{Banach}} Spaces},
  author = {Phelps, R.R},
  date = {1974-09},
  journaltitle = {Journal of Functional Analysis},
  shortjournal = {Journal of Functional Analysis},
  volume = {17},
  number = {1},
  pages = {78--90},
  issn = {00221236},
  doi = {10.1016/0022-1236(74)90005-6},
  url = {https://linkinghub.elsevier.com/retrieve/pii/0022123674900056},
  urldate = {2025-03-12},
  langid = {english}
}

@article{Poisson_formula,
  title = {Fourier {{Transforms Of Unbounded Measures}}},
  author = {Argabright, Loren and De Lamarid, Jesús Gil},
  date = {1971-05},
  journaltitle = {Bulletin of the american mathematical society},
  volume = {77},
  number = {3},
  pages = {355--359},
  langid = {english}
}

@book{Pontryagin,
  title = {Topological {{Groups}}},
  author = {{Pontrjagin L}},
  date = {1946},
  url = {http://archive.org/details/in.ernet.dli.2015.89986},
  urldate = {2025-04-17},
  langid = {english}
}

@online{Proper_metric,
  title = {Proper Metrics on Locally Compact Groups, and Proper Affine Isometric Actions on {{Banach}} Spaces},
  author = {Haagerup, Uffe and Przybyszewska, Agata},
  date = {2006-06-30},
  eprint = {math/0606794},
  eprinttype = {arXiv},
  doi = {10.48550/arXiv.math/0606794},
  langid = {english},
  pubstate = {prepublished}
}

@incollection{qualitative_uncertainty_principle,
  title = {A Qualitative Uncertainty Principle for Locally Compact {{Abelian}} Groups},
  booktitle = {Miniconference on {{Harmonic Analysis}} and {{Operator Algebras}}},
  author = {Hogan, Jeffrey A.},
  date = {1987-01-01},
  volume = {16},
  pages = {133--143},
  publisher = {Australian National University, Mathematical Sciences Institute},
  url = {https://projecteuclid.org/ebooks/proceedings-of-the-centre-for-mathematics-and-its-applications/Miniconference-on-Harmonic-Analysis-and-Operator-Algebras/chapter/A-qualitative-uncertainty-principle-for-locally-compact-Abelian-groups/pcma/1416336207},
  urldate = {2025-04-17}
}

@book{Reed_and_Simon,
  title = {Functional {{Analysis}}},
  author = {Reed, Michael and Simon, Barry},
  date = {1980-01-11},
  edition = {Revised edition},
  publisher = {Academic Press},
  location = {San Diego, Calif.},
  isbn = {978-0-12-585050-6},
  langid = {english},
  pagetotal = {416}
}

@article{Rihaczek,
  title = {Signal Energy Distribution in Time and Frequency},
  author = {Rihaczek, A.},
  date = {1968-05},
  journaltitle = {IEEE Transactions on Information Theory},
  volume = {14},
  number = {3},
  pages = {369--374},
  issn = {1557-9654},
  doi = {10.1109/TIT.1968.1054157},
  url = {https://ieeexplore.ieee.org/abstract/document/1054157},
  urldate = {2025-07-30}
}

@book{Rudin,
  title = {Fourier Analysis on Groups},
  author = {Rudin, Walter},
  date = {1996},
  series = {Wiley {{Classics Library}}},
  edition = {Nachdr},
  publisher = {Wiley-Interscience},
  location = {New York, NY},
  isbn = {978-0-471-52364-2},
  langid = {english},
  pagetotal = {285}
}

@article{Stone_VonNeumann_Mackey,
  title = {An Easy Proof of the {{Stone}}–von {{Neumann}}–{{Mackey Theorem}}},
  author = {Prasad, Amritanshu},
  date = {2011},
  journaltitle = {Expositiones Mathematicae},
  shortjournal = {Expositiones Mathematicae},
  volume = {29},
  number = {1},
  pages = {110--118},
  issn = {07230869},
  doi = {10.1016/j.exmath.2010.06.001},
  url = {https://linkinghub.elsevier.com/retrieve/pii/S0723086910000344},
  urldate = {2025-04-17},
  langid = {english}
}

@article{Sudarshan,
  title = {Equivalence of {{Semiclassical}} and {{Quantum Mechanical Descriptions}} of {{Statistical Light Beams}}},
  author = {Sudarshan, E. C. G.},
  date = {1963-04-01},
  journaltitle = {Physical Review Letters},
  shortjournal = {Phys. Rev. Lett.},
  volume = {10},
  number = {7},
  pages = {277--279},
  publisher = {American Physical Society},
  doi = {10.1103/PhysRevLett.10.277},
  url = {https://link.aps.org/doi/10.1103/PhysRevLett.10.277},
  urldate = {2025-06-06}
}

@article{uncertainty_principle_entropy,
  title = {An Entropy-Based Uncertainty Principle for a Locally Compact Abelian Group},
  author = {Özaydin, Murad and Przebinda, Tomasz},
  date = {2004-10-01},
  journaltitle = {Journal of Functional Analysis},
  shortjournal = {Journal of Functional Analysis},
  volume = {215},
  number = {1},
  pages = {241--252},
  issn = {0022-1236},
  doi = {10.1016/j.jfa.2003.11.008},
  url = {https://www.sciencedirect.com/science/article/pii/S0022123603004166},
  urldate = {2025-04-17}
}

@article{vourdasHarmonicAnalysisRational2012,
  title = {Harmonic Analysis on Rational Numbers},
  author = {Vourdas, A.},
  date = {2012-10},
  journaltitle = {Journal of Mathematical Analysis and Applications},
  shortjournal = {Journal of Mathematical Analysis and Applications},
  volume = {394},
  number = {1},
  pages = {48--60},
  issn = {0022247X},
  doi = {10.1016/j.jmaa.2012.04.059},
  langid = {english}
}

@book{Weil,
  title = {L'intégration dans les groupes topologiques et ses applications},
  author = {Weil, André},
  date = {1951},
  eprint = {telUAAAAYAAJ},
  eprinttype = {googlebooks},
  publisher = {Hermann},
  isbn = {978-2-7056-1145-3},
  langid = {french},
  pagetotal = {168}
}

@article{Wigner,
  title = {On the {{Quantum Correction For Thermodynamic Equilibrium}}},
  author = {Wigner, E.},
  date = {1932-06-01},
  journaltitle = {Physical Review},
  shortjournal = {Phys. Rev.},
  volume = {40},
  number = {5},
  pages = {749--759},
  issn = {0031-899X},
  doi = {10.1103/PhysRev.40.749},
  url = {https://link.aps.org/doi/10.1103/PhysRev.40.749},
  urldate = {2025-04-02},
  langid = {english}
}
\end{document}